\documentclass[12pt,draftcls,journal,onecolumn]{IEEEtran}
\usepackage{amsfonts,color,morefloats,pslatex}
\usepackage{amssymb,amsthm, amsmath,latexsym}

\newtheorem{theorem}{Theorem}
\newtheorem{lemma}[theorem]{Lemma}

\newtheorem{open}[theorem]{Open Problem}

\newtheorem{example}[theorem]{Example}

\newtheorem{remark}[theorem]{Remark}

\newcommand{\ord}{{\mathrm{ord}}}

\newcommand{\lcm}{{\mathrm{lcm}}}
\newcommand{\tr}{{\mathrm{Tr}}}

\newcommand{\gf}{{\mathrm{GF}}}
\newcommand{\PG}{{\mathrm{PG}}}

\newcommand{\support}{{\mathrm{suppt}}}

\newcommand{\wt}{{\mathtt{wt}}}

\newcommand{\cV}{{\mathcal{V}}}
\newcommand{\cT}{{\mathcal{T}}}
\newcommand{\C}{{\mathcal{C}}}

\usepackage{blindtext}

\ifCLASSINFOpdf

\else

\fi

\begin{document}
%
\title{Infinite families of cyclic and negacyclic codes  supporting 3-designs\thanks{
X. Wang's research was supported by The National Natural Science Foundation of China under Grant Number  12001175.
C. Tang's research was supported by The National Natural Science Foundation of China under Grant Number  11871058.
C. Ding's research was supported by The Hong Kong Research Grants Council, Proj. No. $16301522$.
}
}

\author{
Xiaoqiang Wang,\thanks{X. Wang is with The Hubei Key Laboratory of Applied Mathematics, Faculty of Mathematics and Statistics, Hubei University, Wuhan 430062, China
(email: waxiqq@163.com)}
Chunming Tang,\thanks{C. Tang is with The School of Mathematics and Information, China West Normal University,  Nanchong, Sichuan,  637002, China (email: tangchunmingmath@163.com)}
and
Cunsheng Ding \thanks{C. Ding is with The Department of Computer Science
                           and Engineering, The Hong Kong University of Science and Technology,
Clear Water Bay, Kowloon, Hong Kong, China (email: cding@ust.hk)}
}

\maketitle

\begin{abstract}
Interplay between coding theory and combinatorial $t$-designs has been a hot topic for many years for combinatorialists and coding theorists. Some infinite families of cyclic codes supporting infinite families of $3$-designs have been constructed  in the
past 50 years. However, no infinite family of negacyclic codes supporting an infinite family of $3$-designs has been reported in the literature. This is the main motivation of this paper.
Let $q=p^m$, where $p$ is an odd prime and $m \geq 2$ is an integer.
The objective of this paper is to present an infinite family of cyclic codes over $\gf(q)$ supporting an infinite family of $3$-designs and two infinite families of negacyclic codes over $\gf(q^2)$ supporting two infinite families of $3$-designs.
The parameters and the weight distributions of these codes are determined.
The subfield subcodes of these negacyclic codes  over $\gf(q)$ are studied.
Three infinite families of almost MDS codes are also presented.
A constacyclic code over GF($4$) supporting a $4$-design and six open problems are also presented in this paper.
\end{abstract}

\begin{IEEEkeywords}
Constacyclic code, \and cyclic code, \and negacyclic code, \and t-design, \and Steiner system.
\end{IEEEkeywords}

%
\IEEEpeerreviewmaketitle

\section{Introduction}\label{sec-intro}

Let $\kappa$ and $v$ be positive integers such that $1\leq \kappa\leq v$. Let $\mathcal{P}$ be a set of $v$ elements and let
$\mathcal{B}$
be a set of $\kappa$-subsets of $\mathcal{P}$.
The pair $\mathbb{D}=(\mathcal{P},\mathcal{B})$ is an  {\it incidence structure}, where the incidence relation is the set membership. The incidence structure $\mathbb{D}=(\mathcal{P},\mathcal{B})$ is called a $t$-$(v,\kappa,\lambda)$
design, or simply $t$-design, if every $t$-subset of $\mathcal{P}$ is contained in exactly $\lambda$ elements of $\mathcal{B}$. The elements of $\mathcal{P}$ are called points, and those of $\mathcal{B}$ are referred to as blocks. The set $\mathcal{B}$ is called the block set of the design. A $t$-design is said to be {\it simple} if $\mathcal{B}$ does not contain any repeated blocks.  A $t$-$(v,\kappa,\lambda)$ design is called a {\it Steiner system} if $t\geq2$ and $\lambda=1$, and is denoted by $S(t,\kappa,v)$.
A $t$-$(v,\kappa,\lambda_t)$ design is also an $s$-$(v,\kappa,\lambda_s)$ design for any integer $s$ with $1 \leq s \leq t$, where
$$\lambda_s=\lambda_t\binom{v-s}{t-s}/ \binom{\kappa-s}{t-s}.$$
Let $b$ denote the number of blocks in $\mathcal{B}$.
The parameters of a $t$-$(n,\kappa,\lambda)$ design satisfy the following equation:
\begin{equation}\label{eq:kbt}
\binom{n}{t}\lambda=\binom{\kappa}{t}b.
\end{equation}

Let $q$ be a prime power and let $\gf(q)$ denote the finite field with $q$ elements. An $[n,k,d]$ code $\C$ over $\gf(q)$ is a $k$-dimensional linear subspace of $\gf(q)^n$ with minimum Hamming distance $d$. Let $A_i$ denote the number of  codewords with Hamming weight $i$ in $\mathcal{C}$. The {\it weight enumerator} of $\mathcal{C}$ is defined by $1+A_1x+A_2x^2+\cdots+A_nx^n$. The {\it weight distribution} of $\mathcal{C}$ is defined by the sequence $(1, A_1, A_2, \ldots, A_n)$. A linear code with parameters $[n, k, n-k+1]$ is said to be MDS (maximum distance separable) and
a linear code with parameters $[n, k, n-k]$ is said to be almost MDS.

Let $\lambda$ be a nonzero element in $\gf(q)$.
A linear code $\C$ of length $n$ over $\gf(q)$ is called a \emph{$\lambda$-constacyclic code} if
$ (c_0,c_1,\ldots,c_{n-1})\in \C$ implies that  $(\lambda c_{n-1},c_0,\ldots,c_{n-2}) \in \C$.  When $\lambda=1$ and $\lambda=-1$,
 $\lambda$-constacyclic codes are called cyclic codes and negacyclic codes, respectively. By definition, constacyclic codes are a
 generalization of cyclic codes, and were first studied in 1967 by Berman in \cite{Berman67}.

For each  vector $(c_0,c_1,\ldots,c_{n-1}) \in \gf(q)^n$, define
\begin{eqnarray*}
\Phi((c_0,c_1,\ldots,c_{n-1}))=c_0+c_1x+c_2x^2+\cdots+c_{n-1}x^{n-1} \in \gf(q)[x]/(x^n-\lambda).
\end{eqnarray*}
It is easily seen that a subset $\C$ of $\gf(q)^n$ is a $\lambda$-constacyclic code if and only if
$\Phi(\C)$ is an ideal of the quotient ring $ \gf(q)[x]/(x^n-\lambda)$. We identify $\C$ with $\Phi(\C)$.
It is known that the ring $\gf(q)[x]/(x^n-\lambda)$ is principal. Thus, every $\lambda$-constacyclic code $\C$ can be expressed as $\C=( g(x) )$, where
$g(x)$ is a monic polynomial with the smallest degree and $g(x)$ divides $x^n-\lambda$. The polynomial $g(x)$ is called the {\it generator polynomial} and $h(x)=(x^n-\lambda)/g(x)$ is referred to as the {\em check polynomial} of $\C$.
The roots of $g(x)$ and $h(x)$ are called zeros and nonzeros of $\C$ \cite{Huffman03,Pedersen21,Yang21}.

Linear codes and $t$-designs are closely related. It is well known that a design may yield many codes and a code may give many designs. The reader is referred to, for example,   \cite{Dingbook18,DTarXiv,DT20,DTT20,Tonchev98,Tonchev07,TDing1801} for further information.
A construction of $t$-designs with linear codes goes as follows. Let $\mathcal{C}$ be an $[n,k,d]$ code over the finite field
$\gf(q)$. Let the coordinates of a codeword of $\mathcal{C}$ be indexed by $(0,1,\ldots,n-1)$ and define $\mathcal{P}(\mathcal{C})=\{0,1,\ldots,n-1\}$.
For a codeword $\mathbf{c}=(c_0,c_1,\ldots,c_{n-1})  \in \mathcal{C}$, the {\it support} of $\mathbf{c}$ is defined by
$$\support(\mathbf{c})=\{i\,:\,c_i\neq 0, \,i\in \mathcal{P}(\mathcal{C})\}.$$
Let $\mathcal{B}_w(\mathcal{C})$ denote the set of the supports of all codewords with Hamming weight $w$ in $\mathcal{C}$
without repeated blocks.
If the incidence structure $\mathbb{D}_{w}=(\mathcal{P}(\mathcal{C}),\mathcal{B}_w(\mathcal{C}))$ is a $t$-$(n,w,\lambda)$ design for some positive integers $t$ and $\lambda$, where $1 \leq w \leq n$ and $A_w \neq 0$, we say that $\mathcal{C}$ supports or holds a  $t$-design and the supports of the codewords of weight $w$ in $\C$ hold or support a $t$-design. The reader is
referred to \cite{Dingbook18} for a survey of designs from linear codes.

 Let $\alpha$ be a primitive element of $\gf(q^2)$ and $\beta=\alpha^{q-1}$. Then $\beta$ is a $(q+1)$-th root of unity in $\gf(q^2)$.  In \cite{DT20}, Ding and Tang presented a family of near MDS cyclic codes with nonzeros $\beta$ and $\beta^2$ over $\gf(q)$ which support an infinite family of 3-designs,
 where $q=3^m$ and $m\geq 2$.
Recently  in \cite{XTL21}, Xiang et al. presented an infinite family of cyclic codes with nonzeros $\beta^3$ and $\beta^4$ over $\gf(q)$
supporting an infinite family of 3-designs, where $q=7^m$ and $m \geq 2$. In this paper, we will investigate an infinite family of cyclic codes over $\gf(q)$ with nonzeros $\beta^{\frac{p^s-1}{2}}$ and $\beta^{\frac{p^s+1}{2}}$ and prove that these cyclic codes  support 3-designs, where $q=p^m$, $p$ is an odd prime, $m \geq 2$ is an integer and $s$ is an integer with $1 \leq s <m$. When $(s, p)=(1, 3)$
and  $(s, p)=(1, 7)$, our results cover the main results in~\cite{DT20} and \cite{XTL21}, respectively. Hence, our results about
this family of cyclic codes are a generalization of the main results in~\cite{DT20} and \cite{XTL21}.

Up to now, a small number of infinite families of cyclic codes holding infinite families of 3-designs have been constructed
\cite{Dingbook18}. However, the authors are aware of only one infinite family of $\lambda$-constacyclic codes supporting
$3$-designs with $\lambda \neq 1$ \cite[p. 355]{Dingbook18}, and the authors have not seen an infinite family of negacyclic codes
supporting $3$-designs.  Motivated by these, we present two infinite families of negacyclic codes over $\gf(q^2)$ supporting infinite
families of $3$-designs and study their subfield subcodes. The objective of this paper is not to construct $3$-designs with new
parameters, but to present infinite families of cyclic and negacyclic codes supporting $3$-designs, as linear codes supporting
$3$-designs must have a high level of regularity (see Section \ref{sec-moc} for the definition of regularity) and have an interesting application in cryptography \cite{YD06}.

The rest of this paper is organized as follows. Section~\ref{sec-prelimilery} recalls some notation and basics of linear codes and combinatorial $t$-deigns.
Section~\ref{sec-cyclic-codes} studies the parameters of a family of cyclic codes over $\gf(q)$ and their duals, and induces some infinite families of $3$-designs supported by these codes. Section~\ref{sec-negacyclic-codes} investigates the parameters of two families of negacyclic codes over $\gf(q^2)$ and their duals, and presents some infinite families of $3$-designs supported by these codes.
Section~\ref{sec-finals} concludes this paper and makes some concluding remarks.

\section{Preliminaries}\label{sec-prelimilery}
In this section, we briefly introduce some known results about linear codes and $t$-designs, which will be used later in this paper.

\subsection{Notation used starting from now on}
Starting from now on, we adopt the following notation unless otherwise stated:
\begin{itemize}
\item $\gf(q)$ is the finite field with $q$ elements, where $q$ is an odd prime power.
\item $U_{u(q^{v}+1)}$ denotes the set of all $u(q^{v}+1)$-th roots of unity in $\gf(q^{2v})$, where $u,v\in \{1,2\}$.
\item $\beta$ is a primitive $(q+1)$-th root of unity in $\gf(q^2)$ and $g_i(x)$ denotes the minimal polynomial of $\beta^i$ over $\gf(q)$.
\item $\delta$ is a primitive $2(q^2+1)$-th root of unity in $\gf(q^4)$ and $g'_i(x)$ denotes the minimal polynomial of $\delta^i$ over $\gf(q^2)$.
\item $v_2(\cdot)$ is the 2-adic order function.
\item ${\rm Tr}_{q^a/q^b}(\cdot)$ is the trace function from  $\gf(q^a)$ to  $\gf(q^b)$, where $b\,|\,a$.
\end{itemize}

\subsection{The trace representation of constacyclic codes over finite fields}

Constacyclic codes and their duals have the following relation.

\begin{lemma}\cite{D10,KS90}\label{lem-sdjoin2}
The dual code of an $[n,k]$ $\lambda$-constacyclic code $\mathcal{C}$ generated by $g(x)$ is an $[n,n-k]$ $\lambda^{-1}$-constacyclic code $\mathcal{C}^{\perp}$ generated by  $\widehat{h}(x)=h_0^{-1}x^kh(x^{-1})$, where $h(x)=(x^n-\lambda)/g(x)$ is the check polynomial of $\mathcal{C}$ and $h_0$ is the coefficient of $x^0$ in $h(x)$.
\end{lemma}

Let $\mathbb{Z}_N$ denote the ring of integers modulo $N$ and assume that $\gcd(N, q)=1$.
Let $h$ be an integer with $0\leq h<N$. The {\it $q$-cyclotomic coset} modulo $N$ of $h$ is defined by
$$C_h^{(q, N)}=\{h,hq,hq^2,\ldots,hq^{\ell_{h}-1}\}\,\, {\text\,\,{\rm mod} \,\,N\subseteq \mathbb{Z}_N, }$$
where $\ell_h$ is the smallest positive integer such that $h\equiv hq^{\ell_h} \pmod N$, and is the size of the $q$-cyclotomic
coset $C_h^{(q, N)}$. The trace representation of $\lambda$-constacyclic codes is documented below.

\begin{lemma}\label{lem-01}\cite[Theorem 1]{SZW20}
Let $\lambda \in {\rm GF}(q)^*$ with $\ord(\lambda)=r$. Let $n$ be a positive integer such that ${\rm \gcd}(n,q)=1$.
Define $m=\ord_{rn}(q)$ and let $\gamma \in  {\rm GF}(q^m)$ be a primitive $rn$-th root of unity such that $\gamma^n=\lambda$. Let $\mathcal{C}$ be a q-ary $\lambda$-constacyclic code of length $n$. Suppose that $\mathcal{C}$ has altogether $s$ pairwise
non-conjugate nonzeros, $\gamma^{i_1}$, $\ldots$, $\gamma^{i_s}$, which are $s$ roots of its check polynomial.
Then $\mathcal{C}$ has the trace representation $\mathcal{C}=\{\mathbf{c}(a_1,a_2,\ldots,a_s)\,:\,a_j\in {\rm GF}(q^{m_j}),\,1\leq j\leq s\}$, where
$$\mathbf{c}(a_1,a_2,\ldots,a_s)=\left(\sum_{j=1}^s{\rm Tr}_{q^{m_j}/q}(a_j\gamma^{-ti_j})\right)_{t=0}^{n-1},$$
$m_j=|C_{i_j}^{(q,rn)}|$ and $C_{i_j}^{(q,rn)}$ is the $q$-cyclotomic coset of $i_j$ modulo $rn$.
\end{lemma}

\subsection{The subfield subcodes of linear codes}

Let $\mathcal{C}$ be an $[n,k]$ code over $\gf(q^h)$, where $q$ is a prime power and $h$ is a positive integer. The subfield subcode of $\mathcal{C}$ over $\gf(q)$, denoted by $\mathcal{C}|_{\gf(q)}$, is
defined by
$$\mathcal{C}|_{\gf(q)}=\mathcal{C}\cap \gf(q)^n.$$

In general, there is no elementary relationship between the dimension of a code and the dimension of its subfield subcode.  Two  bounds on the dimension of a subfield subcode are given by
$$k\geq \dim(\mathcal{C}|_{\gf(q)})\geq n-h(n-k).$$
The following theorem shows a relation between a subfield subcode and a trace code of a linear code, which was developed by Delsarte.
\begin{lemma}[Delsarte Theorem]\label{lem:oct10-05-01}
Let $\mathcal{C}$ be a linear code of length $n$ over $\gf(q^m)$. Then
$$\mathcal{C}|_{\gf(q)}=({\rm Tr}_{q^m/q}(\mathcal{C}^{\perp}))^{\perp}.$$
\end{lemma}

\subsection{Pless power moments}

Let $\mathcal{C}$ be an $[n, k]$ code over $\gf(q)$, and denote its dual by $\mathcal{C}^{\perp}$. Let
 $A_i$ and $A^{\perp}_i$ be the number of codewords of weight $i$ in $\mathcal{C}$ and $\mathcal{C}^{\perp}$, respectively.
 The first four Pless power moments are as follows \cite[p. 259]{Huffman03}:
\begin{equation*}
\begin{split}
&\sum_{i=0}^nA_i=q^k;\\
&\sum_{i=0}^niA_i=q^{k-1}(qn-n-A_1^{\perp});\\
&\sum_{i=0}^ni^2A_i=q^{k-2}[(q-1)n(qn-n+1)-(2qn-q-2n+2)A_1^{\perp}+2A_2^{\perp}];\\
&\sum_{i=0}^ni^3A_i=q^{k-3}[(q-1)n(q^2n^2-2qn^2+3qn-q+n^2-3n+2)-(3q^2n^2-3q^2n-6qn^2+12qn\\
&\hskip 48pt+q^2-6q+3n^2-9n+6)A_1^{\perp}+6(qn-q-n+2)A_2^{\perp}-6A_3^{\perp}].
\end{split}
\end{equation*}
If $A_1^{\perp}=A_2^{\perp}=A_3^{\perp}=0$, then the fifth Pless power moment is as follows:
\begin{equation*}
\begin{split}
&\sum_{i=0}^ni^4A_i=q^{k-4}[(q-1)n(q^3n^3-3q^2n^3+6q^2n^2-4q^2n+q^2+3qn^3-12qn^2+15qn-6q-n^3\\
&\hskip 48pt+6n^2-11n+6)+24A_4^{\perp}].
\end{split}
\end{equation*}

\subsection{The Assmus-Mattson theorem}

The following theorem, which was developed by Assumus and Mattson in \cite{AM69},
provides a necessary condition for a linear code and its dual to hold simple $t$-designs.

\begin{theorem}[Assmus-Mattson Theorem] Let $C$ be an $[n,k,d]$ code over $\gf(q)$. Let $d^{\perp}$ denote the minimum distance of $C^{\perp}$. Let $w$ be the largest integer satisfying $w\leq n$ and
$$w-\left\lfloor\frac{w+q-2}{q-1}\right\rfloor < d.$$
Define $w^{\perp}$ analogously using $d^{\perp}$. Let $(A_i)_{i=0}^n$ and $(A_i^{\perp})_{i=0}^n$ denote the weight distribution of $C$ and $C^{\perp}$, respectively. Fix a positive integer $t$ with $t<d$, and let $s$ be the number of $i$ with $A_i^{\perp}\neq 0$ for $1 \leq i\leq n-t$. Suppose $s\leq d-t$. Then
\begin{itemize}
\item the codewords of weight $i$ in $C$ hold a simple $t$-design provided $A_i\neq 0$ and $d\leq i\leq w$, and
\item the codewords of weight $i$ in $C^{\perp}$ hold a simple $t$-design provided $A_i^\perp \neq 0$ and $d^{\perp}\leq i\leq \text{min}\{n-t,w^{\perp}\}$.
\end{itemize}
\end{theorem}

The reader is referred to \cite{TDXing2020} and \cite[Chapter 16]{Dingbook18} for a generalization of the Assmus-Mattson theorem.

\subsection{Ovoid codes and their designs}\label{sec-ovoidcodedes}

Let $q>2$ be a prime power. An \emph{ovoid} in the projective space $\PG(3, \gf(q))$ is a set of $q^2+1$ points such that
no three of them are colinear (i.e., on the same line).
A \emph{classical ovoid} $\cV$ can be defined as the set of all points given by
\begin{eqnarray}
\cV=\{(0,0,1, 0)\} \cup \{(x,\, y,\, x^2+xy +ay^2,\, 1): x,\, y \in \gf(q)\},
\end{eqnarray}
where $a \in \gf(q)$ is such that the polynomial $x^2+x+a$ has no root in $\gf(q)$.
Such ovoid is called an \emph{elliptic quadric}\index{quadric}, as the points
come from a non-degenerate elliptic quadratic form.  When $q$ is odd,
every ovoid in $\PG(3, \gf(q))$ is equivalent to the elliptic quadric.

For $q=2^{2e+1}$ with $e \geq 1$, there is an ovoid which is not an elliptic quadric
and is called the \emph{Tits oviod}\index{Tits ovoid}. It is defined by
\begin{eqnarray}
\cT=\{(0,0,1,0)\}\cup \{(x,\,y,\, x^{\sigma} + xy +y^{\sigma+2},\,1): x, \, y \in \gf(q)\},
\end{eqnarray}
where $\sigma=2^{e+1}$. It is known that the Tits ovoid is not equivalent to the
elliptic quadric.

An \emph{ovoid code} is a linear code over $\gf(q)$ with parameters $[q^2+1, 4, q^2]$.
Every ovoid can be used to construct an ovoid code and vice versa.
The next theorem summarizes information on ovoid codes  \cite[Chapter 13]{Dingbook18}.

\begin{theorem} \label{thm-ovoidcodeinf}
Let $q >2$ be a prime power.
\begin{itemize}
\item  Every ovoid code $\C$ over $\gf(q)$ must have parameters $[q^2+1, 4, q^2-q]$ and weight enumerate $1+(q^2-q)(q^2+1)z^{q^2-q}+(q-1)(q^2+1)z^{q^2}$. When $q$ is odd, all ovoid codes over $\gf(q)$ are monomially equivalent.
When $q$ is even, the elliptic quadric code and the Tits ovoid code over $\gf(q)$ are not monomially equivalent.
\item The dual of every  ovoid code $\C$ over $\gf(q)$ must have have parameters $[q^2+1, q^2-3, 4]$.
\item The minimum weight codewords in an ovoid code over $\gf(q)$ support a $3$-$(q^2+1, q^2-q, (q-2)(q^2-q-1))$ design and the complementary design of this design is a Steiner system $S(3, 1+q, 1+q^2)$. When $q$ is odd, all these $3$-$(q^2+1, q^2-q, (q-2)(q^2-q-1))$ designs are isomorphic. The minimum weight codewords in the dual of  an ovoid code over $\gf(q)$ support a $3$-$(q^2+1, 4, q-2)$ design.
\end{itemize}
\end{theorem}

When $q$ is even, there is a cyclic-code construction of the elliptic quadric \cite[Chapter 13]{Dingbook18}.
When $q$ is odd, no cyclic code can be used to construct an ovoid in $\PG(3, \gf(q))$. But there is a $\lambda$-constacyclic
construction of any ovoid in $\PG(3, \gf(q))$ for odd $q$ due to the following theorem, which is known in
the literature.

\begin{theorem}\label{thm-constacodeelliptic}
Let $q>2$ be a prime power, $\alpha$ be a primitive element of $\gf(q^4)$, $\gamma=\alpha^{q+1}$ and
$\lambda=\alpha^{(q^2+1)(q+1)}$. Let $h(x)$ be the minimal polynomial of $\gamma$ over $\gf(q)$, and
let $\C[q]$ denote the $\lambda$-constacyclic code of length $q^2+1$ over $\gf(q)$ with check polynomial $h(x)$. Then
$\C[q]$ is an ovoid code, which is monomially equivalent to the elliptic quadric code.
\end{theorem}

Notice that the $q$ in Theorem \ref{thm-constacodeelliptic} could be even and odd. and
the $\lambda$ defined in Theorem \ref{thm-constacodeelliptic} is a generator of $\gf(q)^*$. Hence,
the code $\C[q]$ is neither cyclic nor negacyclic. In this paper, we will present a construction of a negacyclic
ovoid code over $\gf(q)$ for $q \equiv 3 \pmod{4}$.

\section{A family of cyclic codes  supporting 3-designs}\label{sec-cyclic-codes}

Throughout this section, let $m \geq 2$ be an integer, $s$ be an integer with $1 \leq s \leq m-1$,  $p$ be an odd prime, $q=p^m$, and $n=q+1$. The two $q$-cyclotomic cosets $C_{(p^s-1)/2}^{(q,n)}$ and $C_{(p^s+1)/2}^{(q,n)}$ are disjoint, as they are given by
$$
C_{(p^s-1)/2}^{(q,n)} =\left\{ \frac{p^s-1}{2}, \, q+1-  \frac{p^s-1}{2}\right\}
$$
and
$$
C_{(p^s+1)/2}^{(q,n)} =\left\{ \frac{p^s+1}{2}, \, q+1-  \frac{p^s+1}{2}\right\}.
$$
Recall that $\beta$ is a primitive $n$-th root of unity in $\gf(q^2)$ and $g_i(x)$ is the minimal polynomial of $\beta^i$ over
$\gf(q)$.

Let $\mathcal{C}(\frac{p^s-1}{2},\frac{p^s+1}{2})$ denote the cyclic code of length $n=q+1$ over $\gf(q)$ with check polynomial $g_{\frac{p^s-1}{2}}(x)g_{\frac{p^s+1}{2}}(x)$.
In this section, we study the parameters of $\mathcal{C}(\frac{p^s-1}{2},\frac{p^s+1}{2})$ and its dual, and prove that the code and its dual hold 3-designs.

In order to obtain the possible weights of the codewords in $\mathcal{C}(\frac{p^s-1}{2},\frac{p^s+1}{2})$,
 we will make use of the following lemma.

\begin{lemma}\label{lem-add}
Let $q$ be an odd prime power,
then the following statements hold.
\begin{enumerate}
\item[{\rm (i)}] Let $z \in \gf(q^2)\setminus \gf(q)$ be fixed, then
$$U_{q+1}\setminus \{1\}=\left\{ \frac{z+u}{z-u}\, :\, u \in \gf(q)\right\}.$$
\item[{\rm (ii)}] Let $\beta' \in U_{q+1} \setminus \{\pm 1\}$ be fixed, then
$$U_{q+1}\setminus \{\beta'\}=\left\{ \frac{\alpha'\beta'+1}{\alpha'+\beta'}\, :\, \alpha' \in \gf(q)\right\}.$$
\end{enumerate}
\end{lemma}

We will need also the next lemma.

\begin{lemma}\label{conj-21march338}
Let $(a,b,c,d)
\in \gf(q^2)^4\setminus \{(0,0,0,0)\}$. When $y$ runs over $U_{q+1}$, the possible number of the solutions of the equation
\begin{equation}\label{eq:sol}
ay+by^{p^s}+cy^{p^s+1}+d=0
\end{equation}
is $0$, $1$, $2$ or $p^{\gcd(s,m)}+1$.
\end{lemma}

\begin{proof}
If the equation in (\ref{eq:sol}) does not have more than two solutions, there is nothing to prove. We now assume that the equation in (\ref{eq:sol}) has more than two solutions. Let $\beta'$ be a solution of the equation in (\ref{eq:sol}) and $\beta'\neq \pm 1$. From Lemma \ref{lem-add}, we know that any element in $U_{q+1}\setminus \{\beta'\}$ can be expressed as $\frac{\alpha'\beta'+1}{\alpha'+\beta'}$, where $\alpha' \in \gf(q)$. Then the equation in (\ref{eq:sol}) becomes
\begin{equation}\label{add-new}
a\left(\frac{\alpha'\beta'+1}{\alpha'+\beta'}\right)+b\left(\frac{\alpha'\beta'+1}{\alpha'+\beta'}\right)^{p^s}
+c\left(\frac{\alpha'\beta'+1}{\alpha'+\beta'}\right)^{p^s+1}+d=0,
\end{equation}
where $\beta'$ is a fixed constant and $\alpha'$ runs over $\gf(q)$.
Multiplying by  $(\alpha'+\beta')^{p^s+1}$ both sides of the equation in (\ref{add-new}), we have
\begin{equation}\label{conj1}
\begin{split}
0=&(a\beta'+b \beta'^{p^s}+c\beta'^{p^s+1}+d)\alpha'^{p^s+1}+(a+b\beta'^{p^s+1}+c\beta'^{p^s}+d\beta')\alpha'^{p^s}\\
+&(a\beta'^{p^s+1}+b+c\beta'+d\beta'^{p^s})\alpha'+(a\beta'^{p^s}+b\beta'+c+b\beta'^{p^s+1})\\
=&(a+b\beta'^{p^s+1}+c\beta'^{p^s}+d\beta')\alpha'^{p^s}
+(a\beta'^{p^s+1}+b+c\beta'+d\beta'^{p^s})\alpha'+(a\beta'^{p^s}+b\beta'+c+d\beta'^{p^s+1})
\end{split}
\end{equation}
 since $\beta'$ is a solution of the equation in (\ref{eq:sol}). It is clear that a solution $x\in U_{q+1}\setminus \{\beta'\}$ to the equation in (\ref{eq:sol}) corresponds to a solution $\alpha' \in \gf(q)$ of the equation in (\ref{conj1}).
If the coefficient of $\alpha'^{p^s}$ is zero, there is at most one solution to the equation in (\ref{conj1}), which is contradictory to our assumption that the equation in (\ref{eq:sol}) has more than 2 solutions.

Hence, the coefficient of $\alpha'^{p^s}$ is nonzero, which means that the equation in (\ref{conj1}) can be transformed to
\begin{equation}\label{eq:aff}
x^{p^s}+\varphi x=\nu,
\end{equation}
where $\varphi,\nu\in \gf(q^2)$ and $x \in \gf(q)$. Obviously, the corresponding linear equation is
\begin{equation}\label{eq:lin}
x^{p^s}+\varphi x=0.
\end{equation}
Then we can assert that $\varphi \in \gf(q)$ if the equation in (\ref{eq:lin}) has more than one solution in the finite field $\gf(q)$. Hence, the equation in (\ref{eq:lin})
has either exactly one solution or exactly $p^{\gcd(m,s)}$ solutions in $\gf(q)$. From the theory of linearized polynomials, we know that the affine equation in (\ref{eq:aff}) either has no solution or has the same number of solutions as the linear equation in (\ref{eq:lin}). The desired conclusion then follows.
\end{proof}

With the preparations above, we now determine the possible weights of the code $\mathcal{C}(\frac{p^s-1}{2},\frac{p^s+1}{2})$.

\begin{lemma}\label{lem:1019}
Let notation be the same as before. Then $\mathcal{C}(\frac{p^s-1}{2},\frac{p^s+1}{2})$ is a $[q+1, 4]$ cyclic code whose possible nonzero weights are in the set $\left\{q+1,q,q-1,q-p^{\gcd(m,s)}\right\}$.
\end{lemma}

\begin{proof}
Recall that $\beta$ is a $(q+1)$-th root of unity in $\gf(q^2)$.  We then deduce that all the solutions of $x^{q+1}-1=0$ can be expressed as $\beta^i$, where $0\leq i\leq q$.
 It is clear that
$$
(x-\beta^{i})(x-\beta^{qi})=(x-\beta^{i})(x-\beta^{-i})=x^2-(\beta^{i}+\beta^{-i})x+1 =x^2-\tr_{q^2/q}(\beta^{i})x+1  \in \gf(q)[x].
$$
Then the degree of $g_i(x)$ is $1$ or $2$.
It is easy to check that $(q+1)\, \nmid \frac{(p^s-1)(q-1)}{2}$ and $(q+1)\, \nmid \frac{(p^s+1)(q-1)}{2}$, as $1 \leq s <m$.
Then $\beta^{\frac{p^s-1}{2}}$ and $\beta^{\frac{p^s+1}{2}}$ are not elements in $\gf(q)$.
Hence, $$\deg(g_{(p^s-1)/2}(x))=\deg(g_{(p^s+1)/2}(x))=2.$$ This means that the degree of the check polynomial of $\mathcal{C}(\frac{p^s-1}{2},\frac{p^s+1}{2})$ is $4$. Consequently, the dimension of $\mathcal{C}(\frac{p^s-1}{2},\frac{p^s+1}{2})$ is $4$.

It follows from Lemma \ref{lem-01} that the trace expression of $\mathcal{C}(\frac{p^s-1}{2},\frac{p^s+1}{2})$ is given by
\begin{equation*}
\mathcal{C}\left(\frac{p^s-1}{2},\frac{p^s+1}{2}\right)=\left\{\mathbf{c}(a,b)=\left({\rm Tr}_{q^2/q}\left(a\beta^{-\frac{(p^s-1)i}{2}}+b\beta^{-\frac{(p^s+1)i}{2}}\right)_{i=0}^q\right)\,:\,a,b \in \gf(q^2)\right\}.
\end{equation*}
Let $\mu=\beta^{-i}$, then $\mu \in U_{q+1}$. It is easily seen that
\begin{equation*}
\begin{split}
{\rm Tr}_{q^2/q}\left(a\mu^{\frac{(p^s-1)}{2}}+b\mu^{\frac{(p^s+1)}{2}}\right)
&=a\mu^{\frac{p^s-1}{2}}+b\mu^{\frac{p^s+1}{2}}+a^q\mu^{-\frac{p^s-1}{2}}+b^q\mu^{-\frac{p^s+1}{2}}\\
&=\mu^{-\frac{p^s+1}{2}}\left(a\mu^{p^s}+b\mu^{p^s+1}+c^q\mu+b^q\right).
\end{split}
\end{equation*}
Hence, if $(a,b)\neq (0,0)$, from Lemma \ref{conj-21march338} we know that the number of $\mu \in U_{q+1}$ such that $${\rm Tr}_{q^2/q}\left(a\mu^{\frac{(p^s-1)}{2}}+b\mu^{\frac{(p^s+1)}{2}}\right)=0$$ is $0$, $1$, $2$ or $p^{\gcd(m,s)}+1$. As a result, for $(a,b)\neq (0,0)$ we have
$$\wt(\mathbf{c}(a,b))\in \left\{q+1,q,q-1,q-p^{\gcd(m,s)}\right\}.$$
This completes the proof.
\end{proof}

In order to obtain the parameters of the dual of the code $\mathcal{C}(\frac{p^s-1}{2},\frac{p^s+1}{2})$, we need the following lemmas.
The following is a well known result.
\begin{lemma}\label{lemma5}
Let $a$ be odd and $u,v$ be positive integers, then
\[\gcd(a^u+1,a^v-1)=\left\{ \begin{array}{lll}
           {2}, & \, \,\, \text{if $v_2(v)\leq v_2(u)$}, \\
           {a^{\gcd(u,v)}+1}, & \, \,\, \text{if $v_2(v)> v_2(u)$}. \end{array}  \right.\]
Moreover,
\[\gcd(a^u+1,a^v+1)=\left\{ \begin{array}{lll}
           {2}, & \, \,\, \text{if $v_2(v)\neq v_2(u)$}, \\
           {a^{\gcd(u,v)}+1}, & \, \,\, \text{if $v_2(v)= v_2(u)$}. \end{array}  \right.\]
\end{lemma}

\begin{lemma}\label{lem:oct1013-27}
Let $x,y$ be two distinct elements in $U_{q+1}\setminus\{1\}$. Let
\begin{equation}\label{eq:1014}
T=\left|\begin{array}{cccc}
1, &1,&1\\
1,&x,&y\\
1,&x^{p^s},&y^{p^s}
\end{array}\right|.
\end{equation}
When $v_2(s)\leq v_2(m)$, the determinant $T\neq 0$. When $v_2(s)> v_2(m)$ and $x,y \notin U_{p^{\gcd(s,m)}+1}$,
$T\neq 0$.
\end{lemma}

\begin{proof}
Let $\alpha$ be a primitive element of $\gf(p^s)$. Let $h=\lcm(s, 2m)$. Then $\gf(p^s)$ and $\gf(q^2)$ are subfields of $\gf(p^h)$.
Define the bivariate polynomial
$$
F(X, Y)= XY^{p^s}-X^{p^s}Y-(Y^{p^s}-Y)+(X^{p^s}-X) \in \gf(p^h)[X,Y].
$$
It can be verified that
\begin{eqnarray}
F(X, Y)=(Y-X)(1-X)(1-Y)\prod_{i=1}^{(p^s-3)/2}(1+\alpha^iX-(\alpha^i+1)Y)\prod_{i=(p^s+1)/2}^{p^s-1}(1+\alpha^iX-(\alpha^i+1)Y).
\end{eqnarray}
For any $x, y \in \gf(p^h)$, we have then
\begin{eqnarray*}
T &=& \left|\begin{array}{cccc}
1, &1,&1\\
1,&x,&y\\
1,&x^{p^s},&y^{p^s}
\end{array}\right| \\
&=& xy^{p^s}-x^{p^s}y-(y^{p^s}-y)+(x^{p^s}-x) \\
&=& F(x, y) \\
&=& (y-x)(1-x)(1-y)\prod_{i=1}^{(p^s-3)/2}(1+\alpha^ix-(\alpha^i+1)y)\prod_{i=(p^s+1)/2}^{p^s-1}(1+\alpha^ix-(\alpha^i+1)y.
\end{eqnarray*}
It is now clear that $T=0$ if and only if $x=y$, or $x=1$ or $y=1$ or $x$ and $y$ satisfy $1+\alpha^i x=(1+\alpha^i)y$
for some $i$ with $1\leq i\leq \frac{p^s-3}{2}$ or $\frac{p^s+1}{2}\leq i\leq p^s-1$.

Suppose that for some $x$ and $y$ in $U_{q+1}\setminus \{1\}$ we have
\begin{equation}\label{eq:1026}
1+\alpha^ix-(\alpha^i+1)y=0
\end{equation}
for some $i$ with $1\leq i\leq \frac{p^s-3}{2}$ or $\frac{p^s+1}{2}\leq i\leq p^s-1$.
It is clear that
\begin{equation*}
\begin{split}
(1+\alpha^ix-(\alpha^i+1)y)^q&=1+\frac{\alpha^{qi}}{x}-\frac{(\alpha^i+1)^q}{y}=1+\frac{\alpha^{qi}}{x}-
 \frac{(\alpha^{i}+1)^{q+1}}{y(\alpha^{i}+1)}\\
& =1+\frac{\alpha^{qi}}{x}-\frac{(\alpha^i+1)^{q+1}}{1+\alpha^ix}=0.
 \end{split}
\end{equation*}
Hence,
 $$\frac{x+\alpha^{qi}}{x}=\frac{(\alpha^i+1)^{q+1}}{1+\alpha^ix},$$
  which is the same as
$x+\alpha^ix^2+\alpha^{qi}+\alpha^{(q+1)i}x=x(\alpha^{(q+1)i}+1+\alpha^{qi}+\alpha^i),$
which implies that
\begin{equation}\label{eq:1025}
x=\alpha^{(q-1)i}.
\end{equation}
From the equation in (\ref{eq:1025}) we know that $x \in \gf(p^s)$ since $\alpha$ is a primitive element of $\gf(p^s)$.
Then from the equation in (\ref{eq:1026}) we know that $y \in \gf(p^s)$.
Hence, we know that $x,y\in \gf(p^s)$ if $T=0$.
There are two cases.

\noindent {\bf Case 1:} $v_2(s)\leq v_2(m)$. From Lemma \ref{lemma5} we have $x^2=y^2=1$ since $x,y \in U_{q+1}$ and $q=p^m$, which are contradictory to $x\neq y\neq 1$. Hence, in this case, $T\neq 0$.

\noindent {\bf Case 2:} $v_2(s)> v_2(m)$. From Lemma \ref{lemma5} we obtain $x,y \in U_{p^{\gcd(s,m)}+1}$ since $x,y \in U_{q+1}$ and $q=p^m$, which implies that $T\neq 0$ if $x,y \notin U_{p^{\gcd(s,m)}+1}$.
This competes the proof.
\end{proof}

We now determine the parameters of the dual of the code $\mathcal{C}(\frac{p^s-1}{2},\frac{p^s+1}{2})$.

\begin{lemma}\label{lem:oct10-27}
The code $\mathcal{C}(\frac{p^s-1}{2},\frac{p^s+1}{2})^\perp$ has parameters $[q+1,q-3, 4]$.
\end{lemma}

\begin{proof}
 We follow the notation in the proof of Lemma \ref{lem:1019}. From Lemma \ref{lem:1019}, we know that
the dimension of the code $\mathcal{C}(\frac{p^s-1}{2},\frac{p^s+1}{2})^\perp$ is $q-3$.
Let $d^\perp$ be the minimum Hamming distance of the code $\mathcal{C}(\frac{p^s-1}{2},\frac{p^s+1}{2})^\perp$. It is easily seen that $d^\perp \geq 2$. It follows from the Singleton bound that
$d^\perp \leq 5.$

If $d^\perp =5$, then $\mathcal{C}(\frac{p^s-1}{2},\frac{p^s+1}{2})^\perp$ is an MDS code. This means that $\mathcal{C}(\frac{p^s-1}{2},\frac{p^s+1}{2})$ has parameters $[q+1,4,q-2]$, which is contradictory to Lemma \ref{lem:1019}. Hence,
$2\leq d^\perp \leq 4$. In the following, we prove that
$d^\perp \neq 2$ and $d^\perp \neq 3$.

If $d^\perp =2$, by definition we have
\begin{eqnarray}\label{eq:oct1027}
\begin{cases}
1+a_1\beta^{\frac{(p^s-1)i}{2}}=0,  \\
1+a_1\beta^{\frac{(p^s+1)i}{2}}=0
\end{cases}
\end{eqnarray}
for some $a_1\in \gf(q)^*$ and $1\leq i\leq q$. Then $a_1\beta^{i}=1$. Raising both sides of $a_1\beta^{i}=1$ to the $(q-1)$-th power, we obtain
$\beta^{(q-1)i}=1.$
Since $\gcd(q-1,q+1)=2$, we know that $i=\frac{q+1}{2}$.
Then substituting the value of $i=\frac{q+1}{2}$ into the equations in (\ref{eq:oct1027}), we can assert that
\begin{eqnarray*}
\begin{cases}
1+(-1)^{\frac{p^s-1}{2}}a_1=0,  \\
1+(-1)^{\frac{p^s+1}{2}}a_1=0.
\end{cases}
\end{eqnarray*}
Solving the two equations above, we get that $a_1=0$, which is contradictory to the assumption that $a_1 \neq 0$. Hence, $d^\perp  \neq 2$.

If $d^\perp =3$, by definition we have
\begin{eqnarray}\label{eq:oct1019}
\begin{cases}
1+a_1\beta^{\frac{(p^s-1)i_1}{2}}+a_2\beta^{\frac{(p^s-1)i_2}{2}}=0,  \\
1+a_1\beta^{\frac{(p^s+1)i_1}{2}}+a_2\beta^{\frac{(p^s+1)i_2}{2}}=0
\end{cases}
\end{eqnarray}
for somne $a_1,a_2\in \gf(q)^*$ and $1\leq i_1\neq i_2\leq q$. Let $x=\beta^{i_1}$ and $y=\beta^{i_2}$. It is clear that $x$ and $y$ are two distinct elements in
 $U_{q+1}$. Then the equations in (\ref{eq:oct1019}) can be written as
\begin{eqnarray}\label{eq:1013}
\begin{cases}
1+a_1x^{\frac{p^s-1}{2}}+a_2y^{\frac{p^s-1}{2}}=0,  \\
1+a_1x^{\frac{p^s+1}{2}}+a_2y^{\frac{p^s+1}{2}}=0,
\end{cases}
\end{eqnarray}
 where  $x\neq y\neq 1$. Raising to the $q$-th power  both sides of the equations in (\ref{eq:1013}), we obtain
$1+a_1x^{-\frac{p^s-1}{2}}+a_2y^{-\frac{p^s-1}{2}}=0$ and $1+a_1x^{-\frac{p^s+1}{2}}+a_2y^{-\frac{p^s+1}{2}}=0$.
Then we can deduce that \[\left|\begin{array}{cccc}
1, &x^{-\frac{p^s+1}{2}},&y^{-\frac{p^s+1}{2}}\\
1,&x^{-\frac{p^s-1}{2}},&y^{-\frac{p^s-1}{2}}\\
1,&x^{\frac{p^s-1}{2}},&y^{\frac{p^s-1}{2}}
\end{array}\right|=0,\]
which holds if and only if
\begin{equation}\label{eq:1031}
\begin{split}
\left|\begin{array}{cccc}
1, &1,&1\\
1,&x,&y\\
1,&x^{p^s},&y^{p^s}
\end{array}\right|=0.
\end{split}
\end{equation}
Consequently, from Lemma \ref{lem:oct1013-27} we know that $v_2(s)> v_2(m)$ and $x,y \in U_{p^{\gcd(s,m)}+1}$ if there exist $x$, $y$ such that the equations in (\ref{eq:1013}) hold. If $v_2(s)> v_2(m)$ and $x,y \in U_{p^{\gcd(s,m)}+1}$, then
 the equations in (\ref{eq:1013}) become
\begin{eqnarray*}
\begin{cases}
1+a_1+a_2=0,  \\
1+a_1x+a_2y=0
\end{cases}
\end{eqnarray*}
since $(p^s-1)/(p^{\gcd(s,m)}+1)$ is even.
Solving the equations above, we get
$$a_1=\frac{y-1}{x-y}\,\, \text{and}\,\,a_2=\frac{x-1}{y-x}.$$
Since $(p^{\gcd(s,m)}+1)\, |\, (q+1)$, we have $x^q=1/x$ and $y^q=1/y$. Then from $a_1^q=a_1$ and $a_2^q=a_2$, we obtain
$$\frac{1/y-1}{1/x-1/y}=\frac{y-1}{x-y}\,\,\text{and}\,\,\frac{1/x-1}{1/y-1/x}=\frac{x-1}{y-x},$$
which implies that $x=y=1$. It is contradictory to $x\neq y\neq 1$.

Combining the discussions above, we deduce that that $d^\perp =4$. This competes the proof.
\end{proof}

Below we present the main results of this section.

\begin{theorem}\label{theorem1}
Let $\ell=\gcd(m,s)$. Then the following statements hold.
\begin{enumerate}
\item[{\rm (i)}] The code $\mathcal{C}(\frac{p^s-1}{2},\frac{p^s+1}{2})$ over $\gf(q)$ has parameters $[q+1,4,q-p^\ell]$ and weight enumerator
\begin{equation*}
\begin{split}
&1+\frac{(q^4-q^3-q^2+q)}{p^{3\ell}-p^\ell}z^{q-p^\ell}+\frac{(q^2-1)(p^\ell q^2+p^\ell q-2q^2)}{2p^\ell-2}z^{q-1}+\frac{(q^2-1)(q^2-q+p^\ell)}{p^\ell}z^q\\
&+\frac{p^\ell(q^4-q^3-q^2+q)}{2+2p^\ell}z^{q+1}.
\end{split}
\end{equation*}
In addition, the minimum weight codewords of $\mathcal{C}(\frac{p^s-1}{2},\frac{p^s+1}{2})$ support a $3$-$(q+1,q-p^\ell,\lambda)$ design with
$$\lambda=\frac{(q-p^\ell)(q-p^\ell-1)(q-p^\ell-2)}{p^{3\ell}-p^\ell}.$$
The complementary design of this design is a $3$-$(q+1,p^\ell+1,1)$ design, i.e., a Steiner system $S(3,p^\ell+1, q+1)$.

\item[{\rm (ii)}] The code $\mathcal{C}(\frac{p^s-1}{2},\frac{p^s+1}{2})^\perp$ over $\gf(q)$ has parameters $[q+1,q-3,4]$ and
the minimum weight codewords of $\mathcal{C}(\frac{p^s-1}{2},\frac{p^s+1}{2})^\perp$ support a $3$-$(q+1,4,\lambda^\perp)$ design with
$$\lambda^\perp=\frac{p^\ell+4p^{2\ell}-2p^{3\ell}-2p^{4\ell}+p^{5\ell}-2}{(p^{2\ell}-1)^2}.$$
If $p=3$ and $\ell=1$, then $\lambda^\perp=1$, i.e., the minimum weight codewords of $\mathcal{C}(\frac{p^s-1}{2},\frac{p^s+1}{2})^\perp$ support a Steiner system $S(3,4, q+1)$.
\end{enumerate}
\end{theorem}

\begin{proof}
By Lemma \ref{lem:1019}, when $(a,b)$ runs over $\gf(q)^2\setminus \{(0,0)\}$, the possible nonzero weights of the codewords in $\mathcal{C}(\frac{p^s-1}{2},\frac{p^s+1}{2})$ are
$$q+1,\,\,q,\,\,q-1\,\,\text{and}\,\,q-p^\ell.$$

Denote $w_0=q+1$, $w_1=q$, $w_2=q-1$ and $w_3=q-p^\ell$. Let $A_{w_i}$ be the number of the codewords with weight $w_i$ in $\mathcal{C}(\frac{p^s-1}{2},\frac{p^s+1}{2})$, where $0\leq i \leq 3$. From Lemma \ref{lem:oct10-27}, we know that the code $\mathcal{C}(\frac{p^s-1}{2},\frac{p^s+1}{2})$ has
 the minimum distance $4$. From the first four Pless power moments, we then have
\[\left\{ \begin{array}{lll}
\sum_{i=0}^3{w_i}=q^4-1,\\
\sum_{i=0}^3w_iA_{w_i}=q^3(q^2-1),\\
\sum_{i=0}^3w_i^2A_{w_i}=q^4(q^2-1),\\
\sum_{i=0}^3w_i^3A_{w_i}=q^7-q^4-q^3+q^2.
\end{array}\right. \]
Solving this system of equations, we obtain
\[\left\{ \begin{array}{lll}
A_{w_0}=p^\ell(q^4-q^3-q^2+q)/(2+2p^\ell),\\
A_{w_1}=(q^2-1)(q^2-q+p^\ell)/p^\ell,\\
A_{w_2}=(q^2-1)(p^\ell q^2+p^\ell q-2q^2)/(2p^\ell-2),\\
A_{w_3}=(q^4-q^3-q^2+q)/(p^{3\ell}-p^\ell).
\end{array}\right. \]
It follows from the Assmus-Mattson Theorem that the minimum weight codewords in $\mathcal{C}(\frac{p^s-1}{2},\frac{p^s+1}{2})$ support a $3$-$(q+1,q-p^\ell,\lambda)$ design. Since the number of the supports of the minimum weight codewords in $\mathcal{C}(\frac{q^s-1}{2},\frac{q^s+1}{2})$ is
\begin{equation}\label{eq:1027-even}
b=\frac{A_{q-p^\ell}}{q-1}=\frac{(q-1)q(q+1)}{p^{3\ell}-p^\ell}.
\end{equation}
From (\ref{eq:kbt}) we deduce
$$\lambda=\frac{(q-p^\ell)(q-p^\ell-1)(q-p^\ell-2)}{p^{3\ell}-p^\ell}.$$

By definition, we know that the complements of the supports of the minimum weight codewords in $\mathcal{C}(\frac{p^s-1}{2},\frac{p^\ell+1}{2})$ support
a $3$-$(q+1,p^\ell+1,\lambda')$ design and the number of the supports is given in (\ref{eq:1027-even}). Then from the equation in (\ref{eq:kbt}) we deduce
$\lambda'=1$. Hence, the complementary design of the design supported by the minimum weight codewords in $\mathcal{C}(\frac{p^s-1}{2},\frac{p^s+1}{2})$  is a Steiner system $S(3,p^{\ell}+1, q+1)$.

From the fifth Pless power moment, we then obtain the number of codewords with weight $4$ in $\mathcal{C}^(\frac{p^s-1}{2},\frac{p^s+1}{2})^\perp$, which is given by
\begin{equation*}
A_4^{\perp}=\frac{q(q+1)(q-1)^2(p^\ell+4p^{2\ell}-2p^{3\ell}-2p^{4\ell}+p^{5\ell}-2)}{24(p^{2\ell}-1)^2}.
\end{equation*}
Then by the Assmus-Mattson Theorem again,  we deduce that the codewords of weight $4$ in $\mathcal{C}(\frac{p^s-1}{2},\frac{p^s+1}{2})^\perp$ support a $3$-$(q+1,4,\lambda^\perp)$ design with
$$\lambda^\perp=\frac{p^\ell+4p^{2\ell}-2p^{3\ell}-2p^{4\ell}+p^{5\ell}-2}{(p^{2\ell}-1)^2}.$$
This competes the proof.
\end{proof}

\begin{example}
Let $p=3$, $s=3$ and $m=4$. Then $\mathcal{C}(13,14)$ has parameters $[82, 4, 78]$ and weight enumerator
$$1+1771200x^{78}+1158560x^{80}+14176160x^{81}+15940800x^{82}.$$
The dual code $\mathcal{C}(13,14)^\perp$ has parameters $[82, 78,4]$. The codewords of weight $78$ in $\mathcal{C}(13,14)$ support a $3$-$(82,78,19019)$ design, and the codewords of weight $4$ in $\mathcal{C}(13,14)^\perp$ support a $3$-$(82,4,1)$ design.
\end{example}

\begin{example}
Let $p=5$, $s=1$ and $m=3$. Then $\mathcal{C}(2,3)$ has parameters $[126, 4, 120]$ and weight enumerator
$$1+2018100x^{120}+92767500x^{124}+48450024x^{125}+100905000x^{126}.$$
The dual code $\mathcal{C}(2,3)^\perp$ has parameters $[126, 122,4]$. The codewords of weight $120$ in $\mathcal{C}(2,3)$ support a $3$-$(126,120,14042)$ design, and the codewords of weight $4$ in $\mathcal{C}(2,3)^\perp$ support a $3$-$(126,4,3)$ design.
\end{example}

\begin{remark}
The code $\C(\frac{p^m-1}{2}, \frac{p^m+1}{2})$ is an MDS code coder $\gf(p^m)$ with parameters $[p^m+1, 3, p^m-1]$.
\end{remark}

\begin{remark}\label{remark1}
{\em We have the following comments on the codes  $\mathcal{C}(\frac{p^s-1}{2},\frac{p^s+1}{2})$ and their duals.
These comments clarify the contributions of this section and explain why these codes are interesting.
\begin{itemize}
\item We inform that a linear code over $\gf(q)$ with the same parameters as the cyclic code $\mathcal{C}(\frac{p^s-1}{2},\frac{p^s+1}{2})$
was presented in \cite{XCQ22}. But the code in  \cite{XCQ22} is not cyclic. Hence, the code
$\mathcal{C}(\frac{p^s-1}{2},\frac{p^s+1}{2})$ is more interesting, as it is cyclic.
\item
When $(s, p)=(1,3)$, Theorem \ref{theorem1} becomes the main result of \cite{DT20}.
When $(s, p)=(1,7)$, Theorem \ref{theorem1} becomes the main result of \cite{XTL21}.
Hence, Theorem \ref{theorem1} is a generalization of the works in \cite{DT20} and \cite{XTL21}.

\item
In \cite{LDMTT}, using a group-theoretic approach the authors presented a family of BCH codes of length $\delta^m+1$ over $\gf(\delta)$ whose minimum weight
codewords support a Steiner system $S(3, \delta+1, \delta^m+1))$, which is isomorphic to the Witt spherical geometry design
with the same parameters, where $\delta$ is any prime power.
The $3$-design supported by  $\mathcal{C}(\frac{p^s-1}{2},\frac{p^s+1}{2})$ is not a Steiner system, but its complementary
design
is a Steiner system. In this sense, Theorem \ref{theorem1} complements the work of \cite{LDMTT}. Notice that
$\mathcal{C}(\frac{p^s-1}{2},\frac{p^s+1}{2})$ is not a BCH code.

\item The parameters of the $3$-$(q+1, q-p^\ell, \lambda)$ design supported by the minimum weight codewords in
$\mathcal{C}(\frac{p^s-1}{2},\frac{p^s+1}{2})$ are not new, as the complementary design of this design is a Steiner
system $S(3, p^\ell+1, q+1)$. It is open if this design is isomorphic to the complementary design of the Witt spherical geometry
design with parameters $3$-$(q+1, p^\ell +1, 1)$ or the $3$-design presented in \cite{XCQ22}.
\end{itemize}
}
\end{remark}

\begin{open}\label{open-22221}
Is the $3$-design supported by the minimum weight codewords in $\mathcal{C}(\frac{p^s-1}{2},\frac{p^s+1}{2})$
isomorphic to the complementary design of the Witt spherical geometry design with parameters $3$-$(q+1, p^\ell +1, 1)$?
\end{open}

\section{Two families of negacyclic codes supporting 3-designs}\label{sec-negacyclic-codes}\label{sec-negacyclicc}

Throughout this section, let $q=p^m$ and $n=q^2+1$, where $p$ is an odd prime and $m\geq 2$ is a positive integer.
Let $\alpha$ be a primitive element of $\gf(q^4)$ and put $\delta=\alpha^{(q^2-1)/2}$. Then $\delta^n=-1$ and $\delta$
is a primitive $2n$-th root of unity. For each $i$ with $0 \leq i <2n$, let $g'_i(x)$ denote the minimal polynomial of
$\delta^i$ over $\gf(q^2)$. It is easy to check that the following $q^2$-cyclotomic cosets,
\begin{eqnarray*}
C_1^{(q^2, 2n)}  &=&  \{1, q^2\}, \\
C_{q^2+q+1}^{(q^2, 2n)}  &=&  \{q^2+q+1, 2q^2-q+2\}, \\
C_{q^2-q+1}^{(q^2, 2n)}  &=&   \{q^2-q+1, q\},
\end{eqnarray*}
have size $2$ and are pairwise disjoint. It then follows that
\begin{eqnarray*}
g'_1(x) &=& (x- \delta)(x-\delta^{q^2}), \\
g'_{q^2+q+1}(x) &=& (x- \delta^{q^2+q+1})(x-\delta^{2q^2-q+2}), \\
g'_{q^2-q+1}(x) &=& (x- \delta^{q^2-q+1})(x-\delta^{q}).
\end{eqnarray*}
In this section, we consider the two families of negacyclic codes of length $n$ over $\gf(q^2)$ with check polynomials
$g'_1(x) g'_{q^2+q+1}(x) $ and $g'_1(x) g'_{q^2-q+1}(x) $ and their duals. We prove that these codes and their duals
support 3-designs under certain conditions. We also determine the subfield subcodes of these negacyclic codes.

\subsection{The first family of negacyclic codes supporting 3-designs }

 Throughout this subsection, we always assume that $q\equiv 1 \pmod 4$. Let $\mathcal{C}(1,q^2+q+1)$ denote the negacyclic code of length $n=q^2+1$ over $\gf(q^2)$
with check polynomial $g'_1(x)g'_{q^2+q+1}(x)$. In this subsection, we study the parameters of $\mathcal{C}(1,q^2+q+1)$ and its dual, and prove that the code and its dual support 3-designs. The following lemma will be needed later.

\begin{lemma}\label{lem:aaaa1027}
Let $(a_0,a_1,a_2,a_3)
\in \gf(q^2)^4$.  Define
$$N_0=|\{x\,:\,a_0x+a_1x^{q^2}+a_2x^{q^2+q+1}+a_3x^{q^4+q^3+q^2}=0\,\,\text{and}\,\,x\in U_{2(q^2+1)}\}|$$
and
$$N_1=|\{y\,:\,a_0y+a_1y^q+a_2y^{q+1}+a_3=0\,\,\text{and}\,\,y\in U_{q^2+1}\}|.$$
Then $N_0=2N_1$.
\end{lemma}

\begin{proof}
Since $x\in U_{2(q^2+1)}$, we have
$$x^{q^4+q^3+q^2}=x^{(q^4+q^3+q^2) \bmod 2(q^2+1)} =x^{-q}.$$
Then
\begin{equation}\label{eq:1028}
a_0x+a_1x^{q^2}+a_2x^{q^2+q+1}+a_3x^{q^4+q^3+q^2}=x^{-q}\left(a_0x^{q+1}+a_1x^{q^2+q}+a_2x^{q^2+2q+1}+a_3\right).
\end{equation}

Let $y=x^{q+1}$. It follows from $q\equiv 1 \pmod 4$ and Lemma \ref{lemma5}  that $\gcd(2(q^2+1),q+1)=2$. Then
$y$ takes on each element in $U_{q^2+1}$ exactly twice when $x$ runs over all elements in $ U_{2(q^2+1)}$.
Hence, from the equation in (\ref{eq:1028}) we  have
\begin{equation*}
\begin{split}
N_0&=\left|\{x\,:\,a_0y+a_1y^q+a_2y^{q+1}+a_3=0,\,\,x\in U_{2(q^2+1)}\,\,\text{and}\,\,y=x^{q+1}\}\right|\\
&=2\left|\{y\,:\,a_0y+a_1y^q+a_2y^{q+1}+a_3=0\,\,\text{and}\,\,y\in U_{q^2+1}\}\right| \\
&=2N_1.
\end{split}
\end{equation*}
 This completes the proof.
\end{proof}

Notice that the condition $q \equiv 1 \pmod{4}$ is necessary for the statement of Lemma \ref{lem:aaaa1027} to be true.
This explains why we made it clear at the very beginning of this subsection that $q \equiv 1 \pmod{4}$ is assumed throughout this subsection.
We now determine the possible weights of the code $\mathcal{C}(1,q^2+q+1)$.

\begin{lemma}\label{lem:16}
$\mathcal{C}(1,q^2+q+1)$ is a $[q^2+1, 4]$ negacyclic code over $\gf(q^2)$ whose possible nonzero weights are in the set $\left\{q^2+1,q^2,q^2-1,q^2-q\right\}$.
\end{lemma}

\begin{proof}
Recall that
 $\delta \in  {\rm GF}(q^4)$ is a primitive $2(q^2+1)$-th root of unity.
At the beginning of Section  \ref{sec-negacyclicc} it was showed that
\begin{eqnarray*}
g'_1(x) &=& (x- \delta)(x-\delta^{q^2}), \\
g'_{q^2+q+1}(x) &=& (x- \delta^{q^2+q+1})(x-\delta^{2q^2-q+2}),
\end{eqnarray*}
which are two distinct divisors of $x^n+1$.
Thus, $\deg(g'_1(x) g'_{q^2+q+1}(x) )=4$, which means that the dimension of the code $\mathcal{C}(1,q^2+q+1)$ is $4$.

We now find out the possible nonzero weights of  the code $\mathcal{C}(1,q^2+q+1)$.
According to Lemma \ref{lem-01}, the trace expression of $\mathcal{C}(1,q^2+q+1)$ is given by
\begin{equation}\label{eq:cot10-5}
\mathcal{C}(1,q^2+q+1)=\left\{\mathbf{c}(a,b)=\left({\rm Tr}_{q^4/q^2}\left(a\delta^{-i}+b\delta^{-(q^2+q+1)i}\right)\right)_{i=0}^{q^2}\,:\,a,b \in \gf(q^4)\right\}.
\end{equation}
Let $x \in U_{2(q^2+1)}$, then
\begin{equation}\label{eq:Tr42}
\begin{split}
{\rm Tr}_{q^4/q^2}\left(ax+bx^{q^2+q+1}\right)=ax+a^{q^2}x^{q^2}+bx^{q^2+q+1}+b^{q^2}x^{q^2(q^2+q+1)}.
\end{split}
\end{equation}
It is easily seen that the following two conclusions are true:
\begin{itemize}
\item If $\delta'\in U_{2(q^2+1)}$ is a solution of ${\rm Tr}_{q^4/q^2}\left(ax+bx^{q^2+q+1}\right)=0$, then $-\delta'\in U_{2(q^2+1)}$ is also a solution of ${\rm Tr}_{q^4/q^2}\left(ax+bx^{q^2+q+1}\right)=0$.
\item Let $\delta''\in U_{2(q^2+1)}$. Then $\delta''\in \{\delta^i\,:\,0\leq i\leq q^2\}$ if and only if $-\delta'' \notin \{\delta^i\,:\,0\leq i\leq q^2\}$.
\end{itemize}
Hence, from the equations in (\ref{eq:cot10-5}) and (\ref{eq:Tr42})  we can deduce that
\begin{equation*}
wt(\mathbf{c}(a,b))=q^2+1-N_0/2=q^2+1-N_1,
\end{equation*}
where $N_0$ and $N_1$ were defined in Lemma \ref{lem:aaaa1027}.
Combining Lemma \ref{conj-21march338} with Lemma \ref{lem:aaaa1027}, we can obtain the desired possible nonzero weights in the code $\mathcal{C}(1,q^2+q+1)$.
This completes the proof.
\end{proof}

\begin{lemma}\label{conj-21oct1010}
The dual code $\mathcal{C}(1,q^2+q+1)^\perp$ over $\gf(q^2)$ has parameters $[q^2+1,q^2-3,4]$.
\end{lemma}

\begin{proof}
It follows from Lemma \ref{lem:16} that
the dimension of the code $\mathcal{C}(1,q^2+q+1)^\perp$ is $q^2-3$.
Let $d^\perp$ denote the minimum Hamming distance of the code $\mathcal{C}(1,q^2+q+1)^\perp$.
With an analysis similar as in the proof of Lemma \ref{lem:oct10-27}, we know that
$2\leq d^\perp \leq 4.$
Below we prove that $d^\perp \neq 2$ and $d^\perp \neq 3$. We follow the notation in the proof of Lemma \ref{lem:16}.

If $d^\perp =2$, by definition we have
\begin{eqnarray}\label{eq:oct1008}
\begin{cases}
1+a_1\delta^i=0,  \\
1+a_1\delta^{(q^2+q+1)i}=0
\end{cases}
\end{eqnarray}
for some $a_1\in \gf(q^2)^*$ and $1\leq i\leq q^2$. Then $a_1\delta^{(q^2+q)i}=1$. Raising both sides of $a_1\delta^{(q^2+q)i}=1$ to the $(q^2-1)$-th power, we obtain
\begin{equation}\label{eq:1028-01}
\delta^{(q^2+q)(q^2-1)i}=1.
\end{equation}
It is clear that $v_2(2(q^2+1))=2$ and $\gcd(q^2-1,q^2+1)=2$. From Lemma \ref{lemma5}, we know that $\gcd(q+1,q^2+1)=2$.
As a result, we have
$$\gcd\left((q^2+q)(q^2-1),2(q^2+1)\right)=4.$$
Since $\delta^{2(q^2+1)}=1$ and $1\leq i\leq q^2$, from the equation in (\ref{eq:1028-01}) we have $\delta^{4i}=1$ and  $i=\frac{q^2+1}{2}$.  Since $q\equiv 1 \pmod 4$, it is easy to verify that
$$\frac{(q^2+q+1)(q^2+1)}{2}\equiv \frac{3(q^2+1)}{2} \pmod {2(q^2+1)}.$$
 Substituting the value of $i=\frac{q^2+1}{2}$ into the equation in (\ref{eq:oct1008}), we arrive at that
\begin{eqnarray*}
\begin{cases}
1+a_1\delta^{\frac{q^2+1}{2}}=0,  \\
1-a_1\delta^{\frac{q^2+1}{2}}=0.
\end{cases}
\end{eqnarray*}
Solving the two equations above, we get that $a_1=0$, which is contradictory to the assumption that $a_1 \neq 0$. Hence, $d^\perp  \neq 2$.

If $d^\perp =3$, by definition we have
\begin{eqnarray}\label{eq:1010}
\begin{cases}
1+a_1\delta^{i_1}+a_2\delta^{i_2}=0,  \\
1+a_1\delta^{\left(q^2+q+1\right){i_1}}+a_2\delta^{\left(q^2+q+1\right){i_2}}=0
\end{cases}
\end{eqnarray}
for some $a_1,a_2\in \gf(q^2)^*$ and $1\leq i_1\neq i_2\leq q^2$. It is obvious that
 $\delta^{(q^2+1){i_1}}=(-1)^{i_1}$ and $\delta^{(q^2+1){i_2}}=(-1)^{i_2}.$
Then the equations in
(\ref{eq:1010}) becomes
\begin{eqnarray}\label{eq:oct02}
\begin{cases}
1+a_1\delta^{i_1}+a_2\delta^{i_2}=0,  \\
1+(-1)^{i_1} a_1\delta^{qi_1}+(-1)^{i_2}a_2\delta^{qi_2}=0,
\end{cases}
\end{eqnarray}
which implies that
\begin{equation}\label{eq:aadf}
(a_1^q-(-1)^{i_1}a_1)\delta^{qi_1}=-(a_2^q-(-1)^{i_2}a_2)\delta^{q{i_2}}.
\end{equation}
There are the following three cases.

\noindent{\bf Case 1:} $a_1^q-(-1)^{i_1}a_1=0$. In this case, we have $a_1^q=(-1)^{i_1}a_1$, which implies that $a_1^2 \in \gf(q)$. From the equation in (\ref{eq:aadf}), we have $a_2^q=(-1)^{i_2}a_2$ and $a_2^2 \in \gf(q)$.
It then follows from the first equation of (\ref{eq:oct02}) that
\begin{equation}\label{eq110901}
1+a_1\delta^{i_1}=-a_2\delta^{i_2}
\end{equation}
Raising both sides of the equation in (\ref{eq110901}) to the $q^2$-th power, we get that
\begin{equation}\label{eq1109}
1+a_1(-1)^{-i_1}\delta^{-i_1}=-a_2(-1)^{-i_2}\delta^{-i_2}
\end{equation}
since $\delta^{q^2i}=(-1)^{-i}\delta^{-i}$ for any positive integer $i$.
Multiplying both sides of the equation in (\ref{eq110901}) with the corresponding two sides of the equation in (\ref{eq1109}) yields
$$1+a_1\delta^{i_1}+a_1(-1)^{-i_1}\delta^{-i_1}+(-1)^{-i_2}a_1^2=(-1)^{-i_2}a_2^2.$$
Since $a_1^2, a_2^2 \in \gf(q)$, we have $a_1\delta^{i_1}+a_1(-1)^{-i_1}\delta^{-i_1}\in \gf(q)$. Then we deduce that
$$\delta^{2i_1}+\delta^{-2i_1}\in \gf(q)$$
since
$$(a_1\delta^{-i_1}+a_1(-1)^{i_1}\delta^{i_1})^2=a_1^2(\delta^{2i_1}+\delta^{-2i_1}+1)\in \gf(q),$$
which implies that $\delta^{2i_1} \in \gf(q^2)$.  Since $\delta$ is a primitive $2(q^2+1)$-th root of unity, we deduce that
$(q^2+1)\,|\,(q^2-1)i_1,$ which implies that
$i_1=\frac{q^2+1}{2}$ since $\gcd((q^2+1),q(q-1))=2$ and $1\leq i_1\leq q^2$.
Similarly, we can obtain that $i_2=\frac{q^2+1}{2}$, which is contradictory to our earlier assumption that $i_1\neq i_2$.

\noindent{\bf Case 2:} $a_2^q-(-1)^{i_2}a_2=0$. In this case, it follows from (\ref{eq:aadf}) that $a_1^q-(-1)^{i_1}a_1=0$,
which was proved to be impossible in Case 1.

\noindent{\bf Case 3:} $a_1^q-(-1)^{i_1}a_1\neq0$ and $a_2^q-(-1)^{i_2}a_2\neq0$. From the equations in (\ref{eq:oct02}) we have
$$\delta^{q(i_1-i_2)}=-\left(\frac{a_2^q-(-1)^{i_2}a_2}{a_1^q-(-1)^{i_1}a_1}\right).$$
It is easily seen that
$$\left(\frac{a_2^q-(-1)^{i_2}a_2}{a_1^q-(-1)^{i_1}a_1}\right)^q=\left(\frac{a_2-(-1)^{i_2}a_2^q}{a_1-(-1)^{i_1}a_1^q}\right)
=(-1)^{i_1-i_2}\left(\frac{a_2^q-(-1)^{i_2}a_2}{a_1^q-(-1)^{i_1}a_1}\right)$$
 since $a_1,a_2 \in \gf(q^2)^*$. Then
$\delta^{2q(i_1-i_2)(q-1)}=1$. Since  $\delta$ is a primitive $2(q^2+1)$-th root of unity, we deduce that
$(q^2+1)\,|\,q(q-1)(i_1-i_2),$ which implies that
$i_1-i_2=\frac{q^2+1}{2}$ since $\gcd((q^2+1),q(q-1))=2$.
Substituting the value of $i_1=i_2+\frac{q^2+1}{2}$ into the first equation in (\ref{eq:oct02}), we obtain that
$$1+a_1\delta^{i_2}\delta^{\frac{q^2+1}{2}}+a_2\delta^{i_2}=0.$$
Since $q\equiv 1 \pmod 4$, we have $2(q^2+1)\, | \, \frac{(q-1)(q^2+1)}{2}$, which implies that $\delta^{\frac{q^2+1}{2}}\in \gf(q)$. Then
$$\delta^{i_2}=-\frac{1}{a_2+a_1\delta^{\frac{q^2+1}{2}}}\in \gf(q^2),$$
which implies that $\delta^{(q^2-1)i_2}=1$. From $\delta^{2(q^2+1)}=1$ and $1\leq i_2\leq q^2$, we deduce that $i_2=\frac{q^2+1}{2}$. It then follows from $i_1=i_2+\frac{q^2+1}{2}$ that $i_1=q^2+1$, which is contradictory to our earlier assumption that $1\leq i_1 \leq q^2$.

Summarizing the conclusions in the three cases above,  we conclude that $d^\perp \neq3$. Consequently,
$d^\perp = 4$ This completes the proof.
\end{proof}

The following theorem documents the main results of this subsection.

\begin{theorem}\label{theorem2}
Let $q\equiv 1 \pmod 4$. Then the following statements hold.
\begin{enumerate}
\item[{\rm (i)}] The code $\mathcal{C}(1,q^2+q+1)$ over $\gf(q^2)$ has parameters $[q^2+1,4,q^2-q]$ and weight enumerator
\begin{equation*}
\begin{split}
1+(q^5-q)z^{q^2-q}+\frac{(q^4-1)(q-1)q^3}{2}\left(z^{q^2-1}+z^{q^2+1}\right)+\left(q^7-q^5+q^4-q^3+q-1\right)z^{q^2}.
\end{split}
\end{equation*}
Furthermore, the minimum weight codewords of $\mathcal{C}(1,q^2+q+1)$ support a $3$-$(q^2+1,q^2-q,\lambda)$ design with
$$\lambda=(q^2-q-1)(q-2).$$
The complementary design of this design is a Steiner system $S(3,q+1, q^2+1)$.

\item[{\rm (ii)}] The negacyclic code $\mathcal{C}(1,q^2+q+1)^\perp$ over $\gf(q^2)$ has parameters $[q^2+1,q^2-3,4]$ and
the minimum weight codewords of $\mathcal{C}(1,q^2+q+1)^\perp$ support a $3$-$(q^2+1,4,q-2)$ design.
\end{enumerate}
\end{theorem}

\begin{proof}
By Lemma \ref{lem:16}, the possible nonzero weights of the codewords in $\mathcal{C}(1,q^2+q+1)$ are
$q^2+1,\,\,q^2,\,\,q^2-1\,\,\text{and}\,\,q^2-q.$
Denote $w_0=q^2+1$, $w_1=q^2$, $w_2=q^2-1$ and $w_3=q^2-q$. Let $A_{w_i}$ denote the number of the codewords with weight $w_i$ in $\mathcal{C}(1,q^2+q+1)$, where $0\leq i \leq 3$. By Lemma \ref{conj-21oct1010}, the minimum Hamming distance of $\mathcal{C}(1,q^2+q+1)^\perp$
 is $4$. From the first four Pless power moments we then have
\[\left\{ \begin{array}{lll}
\sum_{i=0}^3{w_i}=q^8-1,\\
\sum_{i=0}^3w_iA_{w_i}=q^6(q^4-1),\\
\sum_{i=0}^3w_i^2A_{w_i}=q^8(q^4-1),\\
\sum_{i=0}^3w_i^3A_{w_i}=q^{16}+2q^{14}-2q^{12}-2q^{10}+q^4.
\end{array}\right. \]
Solving this system of equations, we obtain
\[\left\{ \begin{array}{lll}
A_{w_3}=q^5-q,\\
A_{w_0}=A_{w_2}=(q^4-1)(q-1)q^3/2,\\
A_{w_1}=q^7-q^5+q^4-q^3+q-1.\\
\end{array}\right. \]
It follows from the Assmus-Mattson Theorem that the minimum weight codewords in $\mathcal{C}(1,q^2+q+1)$ support a $3$-$(q^2+1,q^2-q,\lambda)$ design. Note that the number of the supports of the minimum weight codewords in $\mathcal{C}(1,q^2+q+1)$ is
\begin{equation}\label{eq:1019-even}
b=\frac{A_{q^2-q}}{q^2-1}=\frac{q^5-q}{q^2-1}=q^3+q.
\end{equation}
Then from (\ref{eq:kbt}) we deduce
$$\lambda=(q^2-q-1)(q-2).$$

By definition, the complements of the supports of the minimum weight codewords in $\mathcal{C}(1,q^2+q+1)$ support
a $3$-$(q^2+1,q+1,\lambda')$ design and the number of the supports is given in (\ref{eq:1019-even}). Then from (\ref{eq:kbt}) we deduce
$\lambda'=1$. Hence, the complementary design of the design supported by the minimum weight codewords in $\mathcal{C}(1,q^2+q+1)$ is a Steiner system $S(3,q+1, q^2+1)$.

From the fifth Pless power moment, we then obtain the number of codewords with weight $4$ in $\mathcal{C}(1,q^2+q+1)^\perp$,
 which is given by
\begin{equation}\label{eq:B3}
A_4^{\perp}=\frac{q^2(q-2)(q^2+1)(q^2-1)^2}{24}.
\end{equation}
Then it follows from the Assmus-Mattson Theorem again that the codewords of weight $4$ in $\mathcal{C}(1,q^2+q+1)^\perp$
 support a $3$-$(q^2+1,4,q-2)$ design.  This completes the proof.
\end{proof}

\begin{example}
Let $q=5$. Then $\mathcal{C}(1,31)$ has parameters $[26, 4, 20]$ and weight enumerator
$$1+3120x^{20}+ 156000 x^{24}+ 75504  x^{25}+ 156000 x^{26}.$$
The dual code $\mathcal{C}(1,31)^\perp$ has parameters $[26, 22,4]$. The codewords of weight $20$ in $\mathcal{C}(1,31)$ support a $3$-$(26,20,57)$ design, and the codewords of weight $4$ in $\mathcal{C}(1,31)^\perp$ support a $3$-$(26,4,3)$ design.
\end{example}

\begin{example}
Let $q=9$. Then $\mathcal{C}(1,91)$ has parameters $[82, 4, 72]$ and weight enumerator
$$1+59040x^{72}+19128960x^{80}+4729760x^{81}+19128960x^{82}.$$
The dual code $\mathcal{C}(1,91)^\perp$ has parameters $[82, 78,4]$. The codewords of weight $72$ in $\mathcal{C}(1,91)$ support a $3$-$(82,72,497)$ design, and the codewords of weight $4$ in $\mathcal{C}(1,91)^\perp$ support a $3$-$(82,4,7)$ design.
\end{example}

The following theorem documents  the subfield subcode of $\mathcal{C}(1,q^2+q+1)$ over $\gf(q)$, which is denoted by
$\mathcal{C}(1,q^2+q+1)|_{\gf(q)}$.

\begin{theorem}\label{theorem3}
It holds that $\mathcal{C}(1,q^2+q+1)|_{\gf(q)}=\{\mathbf{0}\}$, i.e., the zero code.
\end{theorem}

\begin{proof}
By the definition of $\mathcal{C}(1,q^2+q+1)$, it was shown earlier that $\delta$, $\delta^{q^2}$, $\delta^{q^2+q+1}$ and
$\delta^{2q^2+2-q}$
are all the nonzeros of $\mathcal{C}(1,q^2+q+1)$. Then by Lemma \ref{lem-sdjoin2}, all zeros of $\mathcal{C}(1,q^2+q+1)^\perp$ are $\delta^{2q^2+1}$, $\delta^{q^2+2}$, $\delta^{q^2-q+1}$ and $\delta^{q}$. According to Lemma \ref{lem-01}, the nonzeros of $\mathcal{C}(1,q^2+q+1)^\perp$ are the nonzeros of ${\rm Tr}_{q^2/q}(\mathcal{C}(1,q^2+q+1)^\perp)$. Then the possible zeros of ${\rm Tr}_{q^2/q}(\mathcal{C}(1,q^2+q+1)^\perp)$ are $\delta^{2q^2+1}$, $\delta^{q^2+2}$, $\delta^{q^2-q+1}$ or $\delta^{q}$. By Lemma \ref{lem:oct10-05-01} we know that
$$\mathcal{C}\left(1,q^2+q+1\right)\big|_{\gf(q)}=\left({\rm Tr}_{q^2/q}\left(\mathcal{C}\left(1,q^2+q+1\right)^\perp\right)\right)^{\perp}.$$
Hence, in order to obtain the desired result, we only need to show that $\delta^{2q^2+1}$, $\delta^{q^2+2}$, $\delta^{q^2-q+1}$ and $\delta^{q}$ are nonzeros of ${\rm Tr}_{q^2/q}(\mathcal{C}(1,q^2+q+1)^\perp)$.

Let the $q$-cyclotomic coset of $s$ modulo $2(q^2+1)$ be denoted by $C_s^{(q, 2(q^2+1))}$. Then
$$C_1^{(q, 2(q^2+1))}=\{1,\,\,q,\,\,q^2,\,\,q^2-q+1\}$$
and
$$C_{q^2+2}^{(q, 2(q^2+1))}=\{q^2+2,\,\,q^2+q+1,\,\,2q^2+2-q,\,\,2q^2+1\}.$$
If $\delta^{2q^2+1}$ is a zero of ${\rm Tr}_{q^2/q}(\mathcal{C}(1,q^2+q+1)^\perp)$, then from $C_{q^2+2}$ we know that $\delta^{q^2+q+1}$ is also a zero of ${\rm Tr}_{q^2/q}(\mathcal{C}(1,q^2+q+1)^\perp)$, which is contradictory to the fact that $\delta^{q^2+q+1}$ is a nonzero of  $\mathcal{C}(1,q^2+q+1)^\perp$. Similarly, we can show that $\delta^{q^2+2}$, $\delta^{q^2-q+1}$ and $\delta^{q}$ are nonzeros of ${\rm Tr}_{q^2/q}(\mathcal{C}(1,q^2+q+1)^\perp)$.
This completes the proof.
\end{proof}

\begin{remark}\label{remark-negacyccode1}
{\em  We have the following comments on the negacyclic code $\mathcal{C}(1,q^2+q+1)$ and its dual.
The comments below summarize the contributions of this subsection and explain why these codes are
interesting.
\begin{itemize}
\item The condition $q \equiv 1 \pmod{4}$ in Theorem \ref{theorem2} is necessary. The conclusions in Theorem \ref{theorem2}
          are not true in the case $q \equiv 3 \pmod{4}$.
\item
Since $\mathcal{C}(1,q^2+q+1)|_\gf(q)=\{\mathbf{0}\}$, the code $\mathcal{C}(1,q^2+q+1)$ over $\gf(q^2)$ is not a lifted
code over $\gf(q^2)$ of an ovoid code over $\gf(q)$.
\item By definition, any ovoid code over $\gf(q^2)$ has parameters $[q^4+1, 4, q^4-q^2]$. Hence, the negacyclic code
$\mathcal{C}(1,q^2+q+1)$ cannot be an ovoid code. This is also clear from the facts that the negacyclic code
$\mathcal{C}(1,q^2+q+1)$ has four nonzero weights and any ovoid code has two nonzero weights.
\item While any ovoid code meets the Griesmer bound, the negacyclic code
$\mathcal{C}(1,q^2+q+1)$ does not meet the Griesmer bound.
\item As shown above, the code $\mathcal{C}(1,q^2+q+1)$ over $\gf(q^2)$ is different from any ovoid code in several senses in coding theory. However, when $q \equiv 1 \pmod{4}$ the the minimum weight codewords in $\mathcal{C}(1,q^2+q+1)$ support a
$3$-$(q^2+1, q^2-q, (q^2-q^2-1)(q-2))$ design and minimum weight codewords in the elliptic quadric code over $\gf(q)$
support also a $3$-$(q^2+1, q^2-q, (q^2-q^2-1)(q-2))$ design. It is open if the two designs are isomorphic or not.
\item Any ovoid code over $\gf(q)$ corresponds to an ovoid in $\PG(3, \gf(q))$ \cite[Chapter 13]{Dingbook18}. However,
the authors are not aware of the geometric meaning of the negacyclic code
$\mathcal{C}(1,q^2+q+1)$.
\item The codes $\mathcal{C}(1,q^2+q+1)$ and their duals are the first two infinite families of negacyclic codes supporting
an infinite family of $3$-designs in the literature.
\item The codes $\mathcal{C}(1,q^2+q+1)$ may be the first family of negacyclic codes with four-weights in the literature.
\item The duals of the codes $\mathcal{C}(1,q^2+q+1)$ are almost MDS and distance-optimal.
\end{itemize}
}
\end{remark}

\begin{open}\label{open-22222}
Is the $3$-design supported by the minimum weight codewords in  the negacyclic code $\C(1,q^2+q+1)$ over $\gf(q^2)$ for $q\equiv 1 \pmod 4$ isomorphic to the $3$-design supported by  the minimum weight codewords in  the elliptic quadric code
over $\gf(q)$?
\end{open}

\begin{open}\label{open-22223}
Is there an infinite family of negacyclic ovoid codes over $\gf(q)$ for $q \equiv 1 \pmod{4}$?
\end{open}

\subsection{The second family of negacyclic codes supporting 3-designs }

Throughout this subsection, we always let $q$ be an odd prime power with $q\equiv 3 \pmod 4$. Let $\mathcal{C}(1,q^2-q+1)$ denote the negacyclic code of length $q^2+1$ over $\gf(q^2)$
with check polynomial $g'_{1}(x)g'_{q^2-q+1}(x)$. In this subsection, we consider the parameters of  $\mathcal{C}(1,q^2-q+1)$ and its dual, and prove that the code and its dual support 3-designs.  We prove that the subfield subcode  $\mathcal{C}(1,q^2-q+1)|_\gf(q)$ is an  ovoid negacyclic code.  We start with the following lemma.

\begin{lemma}\label{lem:aaaa1029}
Let $(a_0,a_1,a_2,a_3)
\in \gf(q^2)^4$.   Define
$$T_0=|\{x\,:\,a_0x+a_1x^{q^2}+a_2x^{q^2-q+1}+a_3x^{q^4-q^3+q^2}=0\,\,\text{and}\,\,x\in U_{2(q^2+1)}\}|$$
and
$$T_1=|\{y\,:\,a_0y+a_1y^q+a_2y^{q+1}+a_3=0\,\,\text{and}\,\,y\in U_{q^2+1}\}|.$$
Then $T_0=2T_1$.
\end{lemma}

\begin{proof}
Let $x \in U_{2(q^2+1)}$.
Then we have
$$x^{q^4-q^3+q^2}=x^{(q^4-q^3+q^2) \bmod 2(q^2+1)} =x^{q}.$$
Hence, we have
\begin{eqnarray}\label{eq:1029}
\lefteqn{ a_0x+a_1x^{q^2}+a_2x^{q^2-q+1}+a_3x^{q^4-q^3+q^2} } \nonumber \\
&=&  a_0x+a_1x^{q^2}+a_2x^{q^2-q+1}+a_3x^q \nonumber \\
&=& x(a_0+a_1x^{q^2-1}+a_2x^{q^2-q}+a_3x^{q-1}).
\end{eqnarray}

Let $y=x^{q-1}$. It follows from $q\equiv 3 \pmod 4$ and Lemma \ref{lemma5} that $\gcd(2(q^2+1),q-1)=2$.
Thus,  $y$ takes on each element in $U_{q^2+1}$ twice when $x$ runs over all elements in $U_{2(q^2+1)}$.
Hence, from the equation in (\ref{eq:1029}) we  have
\begin{equation*}
\begin{split}
T_0&=\left|\{x\,:\,a_0+a_1y^{q+1}+a_2y^{q}+a_3y=0,\,\,x\in U_{2(q^2+1)}\,\,\text{and}\,\,y=x^{q-1}\}\right|\\
&=2\left|\{y\,:\,a_0y+a_1y^q+a_2y^{q+1}+a_3=0\,\,\text{and}\,\,y\in U_{(q^2+1)}\}\right|=2T_1.
\end{split}
\end{equation*}
This completes the proof.
\end{proof}

With an analysis similar as the proof of Lemma \ref{lem:16}, combining Lemma \ref{conj-21march338} with Lemma \ref{lem:aaaa1029}, one can prove the following lemma.

\begin{lemma}\label{lem:1029}
$\mathcal{C}(1,q^2-q+1)$ is a $[q^2+1, 4]$ negacyclic code over $\gf(q^2)$ whose possible nonzero weights are in the set $\{q^2+1,q^2,q^2-1,q^2-q\}$.
\end{lemma}

With an analysis similar as the proof of Lemma \ref{conj-21oct1010}, from Lemma \ref{lem:1029} one can prove the following result.

\begin{lemma}\label{conj-21oct338}
The dual code $\mathcal{C}(1,q^2-q+1)^\perp$ over $\gf(q^2)$ has parameters $[q^2+1,q^2-3,4]$.
\end{lemma}

With the preparations above, one can prove the following theorem. The proof is similar to that of Theorem \ref{theorem2}
and is omitted here.

\begin{theorem}\label{theorem4}
Let $q\equiv 3 \pmod 4$. Then the following statements hold.
\begin{enumerate}
\item[{\rm (i)}] The code $\mathcal{C}(1,q^2-q+1)$ over $\gf(q^2)$ has parameters $[q^2+1,4,q^2-q]$ and weight enumerator
\begin{equation*}
\begin{split}
1+(q^5-q)z^{q^2-q}+\frac{(q^4-1)(q-1)q^3}{2}(z^{q^2-1}+z^{q^2+1})+(q^7-q^5+q^4-q^3+q-1)z^{q^2}.
\end{split}
\end{equation*}
Furthermore, the minimum weight codewords of $\mathcal{C}(1,q^2-q+1)$ support a $3$-$(q^2+1,q^2-q,\lambda)$ design with
$$\lambda=(q^2-q-1)(q-2).$$
The complementary design of this design is a Steiner system $S(3,q+1, q^2+1)$.

\item[{\rm (ii)}] The code $\mathcal{C}(1,q^2-q+1)^\perp$ over $\gf(q^2)$ has parameters $[q^2+1,q^2-3,4]$ and
the minimum weight codewords of $\mathcal{C}(1,q^2-q+1)^\perp$ support a $3$-$(q^2+1,4,q-2)$ design.
\end{enumerate}
\end{theorem}

\begin{example}
Let $q=7$. Then $\mathcal{C}(1,43)$ has parameters $[50, 4, 42]$ and weight enumerator
$$1+16800x^{42}+2469600x^{48}+808800x^{49}+2469600x^{50}.$$
The dual code $\mathcal{C}(1,43)^\perp$ has parameters $[50, 46,4]$. The codewords of weight $42$ in $\mathcal{C}(1,43)$ support a $3$-$(50,42,205)$ design, and the codewords of weight $4$ in $\mathcal{C}(1,43)^\perp$ support a $3$-$(50,4,5)$ design.
\end{example}

\begin{example}
Let $q=11$. Then $\mathcal{C}(1,111)$ has parameters $[122, 4, 110]$ and weight enumerator
$$1+161040x^{110}+97429200x^{120}+19339440x^{121}+97429200x^{122}.$$
The dual code $\mathcal{C}(1,111)^\perp$ has parameters $[122, 110,4]$. The codewords of weight $110$ in $\mathcal{C}(1,111)$ support a $3$-$(122,110,981)$ design, and the codewords of weight $4$ in $\mathcal{C}(1,111)^\perp$ support a $3$-$(122,4,9)$ design.
\end{example}

We now settle the parameters of the subfield subcode over $\gf(q)$ of the negacyclic code $\mathcal{C}\left(1,q^2-q+1\right)$ over $\gf(q^2)$.

\begin{theorem}\label{theorem5}
The subfield subcode $\mathcal{C}\left(1,q^2-q+1\right)|_{\gf(q)}$ has parameters $[q^2+1,4,q^2-q]$ and weight enumerator
$$1+(q^2-q)(q^2+1)z^{q^2-q}+(q-1)(q^2+1)z^{q^2}.$$
The dual code of $\mathcal{C}\left(1,q^2-q+1\right)|_{\gf(q)}$ has parameters $[q^2+1,q^2-3,4]$.
\end{theorem}

\begin{proof}
By the definition of $\mathcal{C}\left(1,q^2-q+1\right)$, it is easy to verify that $\delta$, $\delta^q$, $\delta^{q^2}$ and $\delta^{q^2-q+1}$
are all nonzeros of $\mathcal{C}\left(1,q^2-q+1\right)$. Then from Lemma \ref{lem-sdjoin2}, all zeros of
$\mathcal{C}\left(1,q^2-q+1\right)^\perp$ are $\delta^{2q^2+1}$, $\delta^{2q^2-q+2}$, $\delta^{q^2+2}$ and $\delta^{q^2+q+1}$.

Since $q\equiv 3 \pmod 4$, it is easy to verify that the $q$-cyclotomic coset $C_{q^2+2}^{(q, 2(q^2+1))}$ modulo $2(q^2+1)$ is given by
$$C_{q^2+2}^{(q, 2(q^2+1))}=\left\{q^2+2,\,q^2+q+1,\,2q^2-q+2,\,2q^2+1\right\}.$$
Then it follows from Lemma \ref{lem-01} that ${\rm Tr}_{q^2/q}\left(\mathcal{C}\left(1,q^2-q+1\right)^\perp\right)$ is a negacyclic code over $\gf(q)$ with generator polynomial
$$\left(x-\delta^{q^2+2}\right)\left(x-\delta^{q^2+q+1}\right)\left(x-\delta^{2q^2-q+2}\right)\left(x-\delta^{2q^2+1}\right).$$
By Lemma \ref{lem:oct10-05-01} we know that
$$\mathcal{C}\left(1,q^2-q+1\right)\big|_{\gf(q)}=\left({\rm Tr}_{q^2/q}\left(\mathcal{C}\left(1,q^2-q+1\right)^\perp\right)\right)^{\perp}.$$
Hence, the dimension of $\mathcal{C}\left(1,q^2-q+1\right)\big|_{\gf(q)}$ is $4$.

From Lemma \ref{lem:1029}, $\mathcal{C}\left(1,q^2-q+1\right)$ has minimum distance $q^2-q$. Then it follows from the definition of subfield subcodes that the minimum distance $d\left(\mathcal{C}\left(1,q^2-q+1\right)|_{\gf(q)}\right)\geq q^2-q$. From the Griesmer bound, we know that $d\left(\mathcal{C}\left(1,q^2-q+1\right)|_{\gf(q)}\right)\leq q^2-q$. Hence, we deduce that $\mathcal{C}\left(q,q^2-q+1\right)|_{\gf(q)}$ has parameters $[q^2+1,4,q^2-q]$. This means that $\mathcal{C}\left(q,q^2-q+1\right)|_{\gf(q)}$ is an ovoid code. The weight distribution of $\mathcal{C}\left(q,q^2-q+1\right)|_{\gf(q)}$ and the parameters of the dual code then follow from those of ovoid codes and their duals. This competes the proof.
\end{proof}

\begin{remark}\label{remark-negacycliccode2}
{\em
We have the following comments on the negacyclic codes $\mathcal{C}(1,q^2-q+1)$ and its duals.
These comments show the contributions of this subsection and explain why the negacyclic codes
$\mathcal{C}(1,q^2-q+1)$ and $\mathcal{C}(1,q^2-q+1)|_\gf(q)$ and their duals are interesting.
\begin{itemize}
\item The condition $q \equiv 3 \pmod{4}$ in Theorem \ref{theorem4} is necessary. Otherwise, the conclusions of this theorem
are not true.
\item It follows from Theorems \ref{thm-ovoidcodeinf} and \ref{theorem5} that the $3$-designs supported by the minimum
weight codewords in the code $\mathcal{C}(1,q^2-q+1)$ and its subfield subcode   $\mathcal{C}(1,q^2-q+1)|_\gf(q)$ are the same
and are isomorphic to the $3$-design supported by the minimum weight codewords in the elliptic quadric code.  A novelty of
  Theorems \ref{thm-ovoidcodeinf} and \ref{theorem5} is that they show that the $3$-design supported by the minimum weight codewords in the elliptic quadric code can also be supported by the negacyclic code  $\mathcal{C}(1,q^2-q+1)$ and its subfield subcode $\mathcal{C}(1,q^2-q+1)|_\gf(q)$ in the case $q \equiv 3 \pmod{4}$.
\item  Another novelty of Theorems \ref{thm-ovoidcodeinf} and \ref{theorem5} is the determination of the weight distribution of
a lifted code over $\gf(q^2)$ of an ovoid code over $\gf(q)$ for $q \equiv 3 \pmod{4}$. It is open if the lifted code over $\gf(q^2)$
of any ovoid code over $\gf(q)$ for $q \equiv 3 \pmod{4}$ has the same weight distribution as the negacyclic code
$\mathcal{C}(1,q^2+q+1)$.
\item The codes $\mathcal{C}(1,q^2-q+1)$ may be the second family of negacyclic codes with four-weights in the literature.
\item The codes $\mathcal{C}(1,q^2-q+1)|_\gf(q)$ may be the first family of negacyclic codes with two-weights in the literature.
\item The duals of the codes $\mathcal{C}(1,q^2-q+1)$ are almost MDS and distance-optimal.
\item The duals of the codes $\mathcal{C}(1,q^2-q+1)|_\gf(q)$ are almost MDS and distance-optimal.
\end{itemize}

}
\end{remark}

Although the dimension and minimum distance of the lifted code over $\gf(q^2)$ of any ovoid code over $\gf(q)$ are
known to be $4$ and $q^2-q$,  according to Remark \ref{remark-negacycliccode2} we still have the next two open problems.

\begin{open}\label{open-22224}
Does  the lifted code over $\gf(q^2)$
of any ovoid code over $\gf(q)$ for $q \equiv 3 \pmod{4}$ have the same weight distribution as the negacyclic code
$\mathcal{C}(1,q^2+q+1)$?
\end{open}

\begin{open}\label{open-22225}
What is the weight distribution of the lifted code over $\gf(q^2)$ of an ovoid code over $\gf(q)$ for $q \equiv 1 \pmod{4}$?
Is it unique?
\end{open}

\section{Summary and concluding remarks}\label{sec-finals}

\subsection{Summary of the works in this paper}
The main works of this paper are the following:
\begin{itemize}
\item We determined the weight distribution of the cyclic code $\mathcal{C}\left(\frac{p^s-1}{2},\frac{p^s+1}{2}\right)$ over $\gf(q)$ and the parameters of its dual. We showed that the code and its dual support $3$-designs. The results
documented in Theorem \ref{theorem1}  generalized the main results in~\cite{DT20,XTL21} for $p=3$ and $p=7$, respectively.

\item We derived  the weight distributions of the negacyclic codes $\C(1,q^2+q+1)$ over $\gf(q^2)$ for $q\equiv 1 \pmod 4$ and
$\C(1,q^2-q+1)$ over $\gf(q^2)$ for $q\equiv 3 \pmod 4$. We settled the parameters of the duals of these negacyclic codes. We proved that these two families of  negacyclic codes and their duals support $3$-designs.  A  novelty of Theorem \ref{theorem5} is to show that these $3$-designs with known parameters can be supported by negacyclic codes.
Moreover, we determined the subfield subcodes over $\gf(q)$ of these negacyclic codes  $\C(1,q^2+q+1)$ and $\C(1,q^2-q+1)$.
\end{itemize}

\subsection{The motivation, the objective and the contributions of this paper}\label{sec-moc}

It is observed that a $t$-design could be supported by different linear codes.
Linear codes supporting $3$-designs must have a high level of regularity and may have special applications (e.g., the application in cryptography mentioned below). Here the level of regularity of a linear code may mean the degree of transitivity or homogeneity of its automorphism group or the largest value of $t$ in the set of the conditions in the Assmus-Mattson theorem satisfied by the linear code. All linear codes supporting $t$-designs with $t \geq 2$ are special in some sense, i,e,, they have
a kind of regularity. It would be a nice contribution to combinatorics to prove that some known or new linear codes can
support $3$-designs with new parameters. It would also be an interesting contribution to coding theory to construct linear codes that can
support $3$-designs with known parameters, as such linear codes must be special and may have special applications.
The objective of this paper is not to obtain $3$-designs with new parameters, but to discover new cyclic codes and negacyclic codes supporting $3$-designs, as such codes must have a high level of regularity and an interesting application in cryptography.
Notice that the reference \cite{LDMTT} has the same objective as this paper.

It is known that cyclic codes have a better algebraic structure than other constacyclic codes, constacyclic codes have
a better algebraic structure than non-constacyclic codes, and their decoding complexities are different.
Thus, cyclic and negacyclic codes are more attractive.
These facts explain why we would search for cyclic and negacyclic codes supporting $3$-designs.

The contributions of this paper to design theory is limited, as the parameters of the $3$-designs presented in this paper are not new.
But it is open if the $3$-designs presented in this paper are isomorphic to some known $3$-designs or not.
A contribution of this paper to design theory is the proof of the fact that certain $3$-designs with known parameters can be supported by negacyclic codes.

The contributions of this paper are mainly in coding theory and an application in cryptography, which are summarized below:
\begin{enumerate}
\item A family of $[q+1, q-3, 4]$ cyclic codes $\C(\frac{p^s-1}{2}, \frac{p^s-1}{2})^\perp$ were discovered. These codes are almost MDS and cyclic.
\item A family of $[q^2+1, q^2-3, 4]$ negacyclic codes $\C(1, q^2+q+1)^\perp$ for $q \equiv 1 \pmod{4}$ were discovered. These codes are almost MDS and negacyclic.
\item A family of $[q^2+1, q^2-3, 4]$ negacyclic codes $\C(1, q^2-q+1)^\perp$ for $q \equiv 3 \pmod{4}$ were discovered. These codes are almost MDS and negacyclic.
\item The first family of negacyclic ovoid codes (i.e., $\C(1, q^2-q+1)|_{\gf(q)}$ for $q \equiv 3 \pmod{4}$) were discovered in the
literature. Notice that the ovoid codes presented in Theorem \ref{thm-constacodeelliptic} are not negacyclic.
\item Three infinite families of four-weight codes (i.e., $\C(\frac{p^s-1}{2}, \frac{p^s-1}{2})$,  $\C(1, q^2+q+1)$ for $q \equiv 1 \pmod{4}$ and $\C(1, q^2-q+1)$ for $q \equiv 3 \pmod{4}$) were reported.  These codes are cyclic or negacyclic.
\end{enumerate}
Since these codes presented in this paper support $3$-designs, they can be used to construct secret sharing schemes with special access structures \cite{YD06}. Hence, these codes
have a nice application in cryptography.

The cyclic code $\C(\frac{p^s-1}{2}, \frac{p^s-1}{2})$ over $\gf(p^m)$ with dimension $4$ is also interesting in the sense that it can be
used to construct a cyclic code over $\gf(p)$ with good parameters, as its trace code
$\tr_{q/p}(\C(\frac{p^s-1}{2}, \frac{p^s-1}{2})$) is a $[p^m+1, 4m]$ cyclic code over $\gf(p)$,
which could be the best cyclic code over $\gf(p)$ with length $p^m+1$ and dimension $4m$.
For example, when $(s, m, p)=(1,2,3)$,
 $\tr_{q/p}(\C(\frac{p^s-1}{2}, \frac{p^s-1}{2}))$ has parameters $[10, 8, 2]$ and is the best
 ternary cyclic code of length $10$ and dimension $8$ according to \cite[Table A.81]{Dingbook15}.
  When $(s, m, p)=(1,3,3)$ ,
 $\tr_{q/p}(\C(\frac{p^s-1}{2}, \frac{p^s-1}{2}))$ has parameters $[28, 12, 8]$ and is the best
 ternary cyclic code of length $28$ and dimension $12$ according to \cite[Table A.91]{Dingbook15}.
 It can be easily proved that the dimension of  $\tr_{q/p}(\C(\frac{p^s-1}{2}, \frac{p^s-1}{2}))$ is $4m$.
 But the minimum distance of $\tr_{q/p}(\C(\frac{p^s-1}{2}, \frac{p^s-1}{2}))$ is open. Linear codes
 of small dimensions over $\gf(p^m)$ could be very interesting in coding theory, as their subfield
 subcodes over $\gf(p)$ or their subfield codes (i.e., their trace codes) over $\gf(p)$ may have a
 large dimension and other interesting parameters.

\subsection{Two interesting open problems}

No linear code supporting a nontrivial simple $t$-design for $t \geq 6$ is known in the literature. No infinite family of
linear codes supporting an infinite family of $5$-designs is reported in the literature. Only a small number of sporadic
linear codes supporting a $5$-design are known \cite{Dingbook18}, but none of them is known to be a cyclic code.
Thus, the highest strength known of a design supported by a cyclic code is $4$.
In \cite{TDing1801, YZ22}, infinite families of cyclic codes of length $2^{2m+1}+1$ over $\gf(2^{2m+1})$ supporting
an infinite family of $4$-designs were presented. Note that
$$
\gcd(2^{2m+1}-1,2^{2m+1}+1)=1.
$$
It was proved in \cite{Hu00} that the two rings $\gf(2^{2m+1})[x]/(x^n-\lambda)$ and $\gf(2^{2m+1})[x]/(x^n-1)$ are
scalar-equivalent, where $\lambda$ is a primitive element of $\gf(2^{2m+1})$. Hence, the infinite families of
$\lambda$-constacyclic codes over $\gf(2^{2m+1})$ that are scalar-equivalent to some of the infinite families of cyclic codes of length $2^{2m+1}+1$ over $\gf(2^{2m+1})$ in \cite{TDing1801, YZ22} support infinite families of $4$-designs.

Since no $\lambda$-constacyclic code over $\gf(2^{2m})$ with $\lambda \neq 1$ supporting a $4$-design was reported
in the literature, we present the following example.

\begin{example} \label{exam-openlast}
Let $w$ be a primitive element of $\gf(2^8)$ with $w^8+w^4+w^3+w^2+1=0$. Let $\lambda=w^{85}$. Then $\lambda$
is a primitive element of $\gf(4)$.   Let $\gamma=w^{5}$. Then $\gamma^{17}=\lambda$. The minimal polynomial of
$\gamma$ over $\gf(4)$, denoted by $h_1(x)$, is given by
$$
h_1(x)=x^4 + \lambda^2 x^3 + x^2 + x + \lambda^2
$$
and the minimal polynomial of
$\gamma^7$ over $\gf(4)$, denoted by $h_7(x)$, is given by
$$
h_1(x)=x^4 + x^3 + \lambda x^2 + \lambda x + \lambda^2.
$$
Let $\C$ denote the $\lambda$-constacyclic code of length $17$ over $\gf(4)$ with check polynomial $h_1(x)h_7(x)$.
Then $\C$ has parameters $[17, 8, 8]$ and weight enumerator
\begin{eqnarray*}
1+ 1530z^8 + 10, 8160z^{10} + 25704z^{12}  + 24480z^{14} + 5661z^{16}.
\end{eqnarray*}
The dual code $\C^\perp$ has parameters $[17, 9, 7]$. By the Assmus-Mattson theorem, the minimum codewords in
$\C$ support a $4$-$(17, 8, 15)$ simple design. But the extended code of $\C$ does not support a $1$-design.
\end{example}

A cyclic code of length $17$ over $\gf(4)$ with the same parameters as the code $\C$ in Example \ref{exam-openlast}
was reported in \cite[Appendix A.1]{Dingbook18}. Note that constacyclic codes contain cyclic codes as a proper subclass.
The following
open problems would be interesting.

\begin{open}\label{open-22226}
Is there a constacyclic code over $\gf(q)$ for some $q$ supporting a simple $5$-design?
\end{open}

\begin{open}\label{open-22227}
Is there an infinite family of constacyclic codes over finite fields supporting an infinite
family of simple $5$-designs?
\end{open}

The reader is cordially invited to attack all the seven open problems (i.e., Open Problems \ref{open-22221}, \ref{open-22222},   \ref{open-22223}, \ref{open-22224}, \ref{open-22225}, \ref{open-22226},  and \ref{open-22227} presented in this paper.


\end{document}